\documentclass{article}
\usepackage[utf8]{inputenc}
\usepackage[margin=0.8in]{geometry}
\usepackage{amsmath,amsfonts,amsthm}
\usepackage{bbm}
\usepackage{comment}
\usepackage{scrextend}
\usepackage{hyperref}
\usepackage{tablefootnote}
\usepackage{graphicx}
\usepackage{xcolor}
\usepackage{authblk}

\newcommand{\eps}{\ensuremath{\varepsilon}}

\newcommand{\cL}{\mathcal{L}}
\newcommand{\cR}{\mathcal{R}}
\newcommand{\bE}{\ensuremath{\mathbb{E}}}
\newcommand{\bN}{\ensuremath{\mathbb{N}}}

\newcommand{\cX}{\ensuremath{\mathcal{X}}}
\newcommand{\cY}{\ensuremath{\mathcal{Y}}}
\newcommand{\tY}{\ensuremath{\widetilde{Y}}}

\newcommand{\bit}{\ensuremath{\{0, 1\}}}

\newcommand{\from}{\ensuremath{ \leftarrow }}
\newcommand{\dist}[2]{\ensuremath{ \Delta \left(#1; #2 \right) }}
\newcommand{\distCond}[3]{\ensuremath{ \Delta \left(#1; #2 \left| #3 \right. \right) }}
\newcommand{\supp}[1]{\ensuremath{\mathbf{supp}({#1})}}
\newcommand{\unif}[1]{\ensuremath{ U_{#1} }}
\newcommand{\minEnt}[1]{\ensuremath{H_\infty(#1)}}
\newcommand{\condMinEnt}[2]{\ensuremath{H_\infty(#1|#2)}}
\newcommand{\avgCondMinEnt}[2]{\ensuremath{\tilde{H}_\infty(#1|#2)}}
\newcommand{\blocks}[2]{\ensuremath{ #1_{1}, \ldots, #1_{#2} }}
\newcommand{\functionBlocks}[3]{\ensuremath{ #3(#1_{1}), \ldots, #3(#1_{#2}) }}
\newcommand{\givenBy}{\ensuremath{\sim}}
\newcommand{\closeTo}[1]{\ensuremath{\approx_{#1}}}
\newcommand{\concat}{\ensuremath{\circ}}
\newcommand{\mutualInfoCond}[3]{\ensuremath{I(#1 ; #2 | #3)}}
\newcommand{\indicate}[1]{\ensuremath{\mathbbm{1}\{{#1}\}}}

\newcommand{\ecc}{\ensuremath{\mathtt{EC}}}

\newcommand{\extractorParam}[5]{\ensuremath{ #1 : [ #2 , #3 \mapsto #4 \sim #5 ] }}

\newcommand{\sext}{\ensuremath{\mathtt{SExt}}}
\newcommand{\mac}{\ensuremath{\mathtt{MAC}}}
\newcommand{\samp}{\ensuremath{\mathtt{samp}}}

\newcommand{\ext}{\ensuremath{\mathtt{ext}}}
\newcommand{\fnmext}{\ensuremath{\mathbf{FNMExt}}}
\newcommand{\twonmext}{\ensuremath{\mathbf{2NMExt}}}
\newcommand{\nmExt}{\ensuremath{\mathbf{2NMExt}}}
\newcommand{\betterTwoNMExt}{\ensuremath{\mathbf{nmRaz}}}

\newcommand{\liExt}{\ensuremath{\mathbf{Li}}}
\newcommand{\razExt}{\ensuremath{\mathbf{Raz}}}
\newcommand{\treExt}{\ensuremath{\mathbf{Tre}}}
\newcommand{\crTreExt}{\ensuremath{\mathbf{crTre}}}
\newcommand{\outerExt}{\ensuremath{\mathbf{E}}}
\newcommand{\innerExt}{\ensuremath{\mathbf{C}}}

\newcommand{\outerOutputLength}{\ensuremath{m}}
\newcommand{\innerOutputLength}{\ensuremath{d}}

\newcommand{\treError}{\ensuremath{\eps_T}}
\newcommand{\razError}[1]{\ensuremath{2^{-(1.5)#1}}}
\newcommand{\liError}{\ensuremath{\eps_L}}
\newcommand{\collisionError}{\ensuremath{\eps_{Collision}}}
\newcommand{\fnmError}{\ensuremath{\eps_{fnm}}}
\newcommand{\twonmError}{\ensuremath{\eps_{tnm}}}
\newcommand{\outerError}{\ensuremath{\delta_{\outerExt}}}
\newcommand{\innerError}{\ensuremath{\delta_{\innerExt}}}

\newcommand{\modifiedTampering}{\ensuremath{T_{f, g}}}

\newcommand{\YLeft}{\ensuremath{Y_\ell}}
\newcommand{\YRight}{\ensuremath{Y_r}}
\newcommand{\YLeftTampered}{\ensuremath{g(Y)_\ell}}

\newcommand{\YLeftLength}{\ensuremath{n_\ell}}
\newcommand{\YRightLength}{\ensuremath{n_r}}
\newcommand{\poly}{\text{poly}}

\newcommand{\onote}[1]{{\color{magenta} \footnotesize(MO: #1)}}

\newtheorem{theorem}{Theorem}
\newtheorem{definition}{Definition}
\newtheorem{corollary}{Corollary}
\newtheorem{claim}{Claim}
\newtheorem{lemma}{Lemma}
\newtheorem{ques}{Question}
\newtheorem{question}{Question}
\newtheorem{remark}{Remark}

\allowdisplaybreaks
\usepackage{multirow}

\title{Extractors: Low Entropy Requirements Colliding With Non-Malleability}
\begin{document}

\author[1]{Divesh Aggarwal}
\author[1]{Eldon Chung}
\author[1]{Maciej Obremski}
\affil[1]{\footnotesize National University of Singapore.\\\texttt{\href{mailto:divesh@comp.nus.edu.sg}{divesh@comp.nus.edu.sg}, \href{mailto:eldon.chung@u.nus.edu}{eldon.chung@u.nus.edu}}, \texttt{\href{mailto:obremski.math@gmail.com}{obremski.math@gmail.com}}}


%
%

\maketitle

\begin{abstract}
Two-source extractors are deterministic functions that, given two independent weak sources of randomness, output a (close to) uniformly random string of bits. Cheraghchi and Guruswami (TCC 2015) introduced two-source non-malleable  extractors that combine the properties of randomness extraction with tamper resilience. Two-source non-malleable  extractors have since then attracted a lot of attention, and have very quickly become fundamental objects in cryptosystems involving communication channels that cannot be fully trusted. Various applications of  two-source non-malleable extractors include in particular non-malleable codes, non-malleable commitments, non-malleable secret sharing, network extraction, and privacy amplification with tamperable memory.

The best known constructions of two-source non-malleable  extractors are due to Chattopadhyay, Goyal, and Li (STOC 2016), Li (STOC 2017), and Li (CCC 2019). All of these constructions require both sources to have min-entropy at least $0.99 n$, where $n$ is the bit-length of each source. 

In this work, we introduce collision-resistant randomness extractors. This allows us to design a compiler that, given a two-source non-malleable extractor, and a collision-resistant extractor, outputs a two-source non-malleable extractor that inherits the non-malleability property from the non-malleable extractor, and the entropy requirement from the collision-resistant extractor. Nested application of this compiler leads to a  dramatic improvement of the state-of-the-art mentioned above. We obtain a construction of a two-source non-malleable extractor where one source is required to have min-entropy greater than $0.8n$, and the other source is required to have only $\text{polylog} (n)$ min-entropy. Moreover, the other parameters of our construction, i.e., the output length, and the error remain comparable to prior constructions.  




\end{abstract}

\section{Introduction}\label{sec:intro}

\paragraph{Two-source extractors.}
The problem of constructing efficient two-source extractors for low min-entropy sources with negligible error has been an important focus of research in pseudorandomness for more than 30 years, with fundamental connections to combinatorics and many applications in computer science. The first non-trivial construction was given by Chor and Goldreich~\cite{CG88} who showed that the inner product function is a low-error two-source extractor for $n$-bit sources with min-entropy $(1/2+\gamma)n$, where $\gamma>0$ is an arbitrarily small constant. A standard application of the probabilistic method shows that (inefficient) low-error two-source extractors exist for polylogarithmic min-entropy. While several attempts were made to improve the construction of~\cite{CG88} to allow for sources with smaller min-entropy, the major breakthrough results were obtained after almost two decades. Raz~\cite{Raz05} gave an explicit low-error two-source extractor where one of the sources must have min-entropy $(1/2+\gamma)n$ for an arbitrarily small constant $\gamma>0$, while the other source is allowed to have logarithmic min-entropy. In an incomparable result, Bourgain~\cite{Bou05} gave an explicit low-error two-source extractor for sources with min-entropy $(1/2-\gamma)n$, where $\gamma>0$ is a small constant. An improved analysis by Lewko~\cite{Lew19} shows that Bourgain's extractor can handle sources with min-entropy $4n/9$.

\paragraph{(Seeded) non-malleable extractors.} The problem of privacy amplification against active adversaries was first considered by Maurer and Wolf~\cite{MW97}. In a breakthrough result, Dodis and Wichs~\cite{DW09} introduced the notion of seeded non-malleable extractors as a natural tool towards achieving a privacy amplification protocol in a minimal number of rounds, and with minimal entropy loss. Roughly speaking, the output of a seeded non-malleable extractor with a uniformly random seed $Y$, and a source $X$ with some min-entropy independent of $Y$, should look uniformly random to an adversary who can tamper the seed, and obtain the output of the non-malleable extractor on a tampered seed. 

More precisely, we require that
\[
    \mathbf{nmExt}(X,Y),\mathbf{nmExt}(X,g(Y)), Y
\approx_\eps U_m, \mathbf{nmExt}(X,g(Y)), Y\;,
\]
where $X$ and $Y$ are independent sources with $X$ having sufficient min-entropy and $Y$ uniformly random, $g$ is an arbitrary tampering function with no fixed points, $U_m$ is uniform over $\bit^m$ and independent of $X, Y$, and $\approx_\eps$ denotes the fact that the two distributions are $\eps$-close in statistical distance (for small $\eps$). 

Prior works have also studied seeded extractors with weaker non-malleability guarantees such as look-ahead extractors~\cite{DW09} or affine-malleable extractors~\cite{AHL16}, and used these to construct privacy amplification protocols. 

\paragraph{Non-malleable two-source extractors.}

A natural strengthening of both seeded non-malleable extractors, and two-source extractors are two-source \emph{non-malleable} extractors.  Two-source non-malleable extractors were introduced by Cheraghchi and Guruswami~\cite{CG17}. Roughly speaking, a function $\nmExt:\bit^n\times\bit^n\to\bit^m$ is said to be a non-malleable extractor if the output of the extractor remains close to uniform (in statistical distance), even conditioned on the output of the extractor inputs correlated with the original sources.
In other words, we require that
\[
    \nmExt(X,Y),\nmExt(f(X),g(Y)),Y
    \approx_\eps U_m, \nmExt(f(X),g(Y)),Y \;.
\]
where $X$ and $Y$ are independent sources with enough min-entropy, $f, g$ are arbitrary tampering functions such that one of $f, g$ has no fixed points. 

The original motivation for studying efficient two-source non-malleable extractors stems from the fact that they directly yield explicit split-state non-malleable codes~\cite{DPW18} (provided the extractor also supports efficient preimage sampling). 

The first constructions of non-malleable codes~\cite{DKO13,ADL18} relied heavily on the (limited) non-malleability of the inner-product two-source extractor. Subsequent improved constructions of non-malleable codes in the split-state model relied on both the inner-product two-source extractor~\cite{ADKO15,AO20}, and on more sophisticated constructions of the two-source non-malleable extractors~\cite{CGL16,Li17,Li19}.  Soon after they were introduced, non-malleable extractors have found other applications such as non-malleable secret sharing~\cite{GK18,ADNOPRS19}.

\paragraph{Connections, and state-of-the-art constructions.} As one might expect, the various notions of extractors mentioned above are closely connected to each other. Li~\cite{Li12} obtained the first connection between seeded non-malleable extractors and two-source extractors based on inner products. This result shows that an improvement of Bourgain's result would immediately lead to better seeded non-malleable extractors, and a novel construction of seeded non-malleable extractors with a small enough min-entropy requirement and a small enough seed size would immediately lead to two-source extractors that only require small min-entropy. However,~\cite{Li12} could only obtain seeded non-malleable extractors for entropy rate above $1/2$. 

In yet another breakthrough result,~\cite{CGL16} obtained a sophisticated construction of seeded non-malleable extractors for polylogarithmic min-entropy. Additionally, they showed that similar techniques can also be used to obtain two-source non-malleable extractors. This immediately led to improved privacy amplification protocols and improved constructions of non-malleable codes in the split-state model. Building on this result, in a groundbreaking work, Chattopadhyay and Zuckerman~\cite{CZ19} gave a construction of two-source extractors with polylogarithmic min-entropy and polynomiallly small error. All of these results have subsequently been improved in~\cite{Li16,BDT17,Coh17,Li17,Li19}. We summarize the parameters of the best known constructions of seeded extractors, two-source extractors, seeded non-malleable extractors, and two-source non-malleable extractors alongside those of our construction in Table~\ref{table:main}. We note here that all prior constructions of two-source non-malleable extractors required both sources to have almost full min-entropy. A recent result~\cite{GSZ21} has not been included in this table since it constructs a weaker variant of a non-malleable two-source extractor (that does not fulfil the standard definition) that is sufficient for their application to network extraction. Even if one is willing to relax the definition to that in~\cite{GSZ21}, the final parameters of our two-source non-malleable extractor are better!

The research over the past few years has shown that non-malleable two-source extractors, seeded non-malleable extractors, two-source extractors, non-malleable codes, and privacy amplification protocols are strongly connected to each other in the sense that improved construction of one of these objects has led to improvements in the construction of others. Some results have made these connections formal by transforming a construction of one object into a construction of another object. For instance, in addition to the connections already mentioned, Ben-Aroya et al.\ \cite{BACDLTS18} adapt the approach of~\cite{CZ19} to show explicit \emph{seeded} non-malleable extractors with improved seed length lead to explicit low-error two-source extractors for low min-entropy.

Also,~\cite{AORSS20} showed that some improvement in the parameters of non-malleable two-source extractor constructions from~\cite{CGL16,Li17,Li19} leads to explicit low-error two-source extractors for min-entropy $\delta n$ with a very small constant $\delta>0$.

\begin{table}[]
    \centering 
    \begin{tabular}{|l|l|l|l|l|l|l|}
    \hline
    Citation                                                                                        & Left Rate                                  & Right Rate                                                & Non-malleability                                          \\\hline
       \textbf{Seeded} \\\hline
        \begin{tabular}[c]{@{}l@{}}\cite{RRV99} Theorem 1\end{tabular}                                & $\text{polylog}(n) / n$                                 & $1$                                  & None                                                 \\\hline
    \begin{tabular}[c]{@{}l@{}}\cite{GUV09} Theorem 4.17\end{tabular}                             & $\log (n) / n$                                 & $1$                                  & None                                                 \\\hline
        \textbf{Seeded, Non-malleable} \\\hline
        \cite{Li12}                                                                                    & $1/2 - \gamma$                             & $1$                                      & Right source                                      \\\hline
    \cite{CRS14}                                                                                    & $1/2 + \gamma$                             & $1$                                      & Right source                                      \\\hline
    \begin{tabular}[c]{@{}l@{}}\cite{DLWZ14} Theorem 1.4\end{tabular}                             & $1/2 + \gamma$                             & $1$                                           & Right source               \\\hline
       \cite{CGL16} & $\log^2 n/n$                     & $1$                                    & Right-source                               \\\hline
    \begin{tabular}[c]{@{}l@{}}\cite{Li17} Theorem 6.2\end{tabular}                               & $\log(n)/n$                                 & $1$                                    & Right source       \\\hline
    \begin{tabular}[c]{@{}l@{}}\cite{Li19}\end{tabular}                              & $\log(n)/n$                               & $1$                                              & Right-source        \\\hline     
    \textbf{Two-source} \\\hline
    \cite{CG88}                                                                                     & $1/2$                                      & $1/2$                                                     & None                                               \\\hline
    \cite{Bou05}                                                                                    & $1/2 - \gamma$                             & $1/2 - \gamma$                                            & None                                               \\\hline
    \cite{Raz05}                                                                                                                & $\log(n)/n$                 & $1/2 + \gamma$                               & None                                      \\\hline

    \textbf{Two-source, Non-malleable} \\\hline
    \cite{CGL16} & $1-\frac{1}{n^\gamma}$                     & $1-\frac{1}{n^\gamma}$                                    & Two-sided                               \\\hline
      \begin{tabular}[c]{@{}l@{}}\cite{Li17}\end{tabular}                              & $(1-\gamma)$                               & $(1-\gamma)$                                              & Two-sided        \\\hline     
    \begin{tabular}[c]{@{}l@{}}\cite{Li19} Theorem 1.11\end{tabular}                              & $(1-\gamma)$                               & $(1-\gamma)$                                              & Two-sided        \\\hline     
    \color{red}{This Work }                                                                                                         &\color{red}{ $\text{polylog}(n) / n$     }              & \color{red}{ $4/5 + \gamma$}                    & \color{red}{Two-sided} \\\hline
\end{tabular}
\caption{In the table, we assume that the left source has length $n$, and $\gamma$ is a very small universal constant that has a different value for different results. Most of the constructions two-source non-malleable extractors including ours allow for $t$-time tampering at the cost of a higher min-entropy requirement. In particular (as described in Remark \ref{multitamper1}, \ref{multitamper2}, and \ref{multitamper3}), for our extractor we require the left source to have min-entropy rate $\text{polylog}(n)/n$, and the right source has min-entropy rate  $(1-\frac{1}{2t+3})$.}
\label{table:main}
\end{table}

Parameters for each extractor were chosen such that the error is $2^{-\kappa^c}$ and the output length is $\Omega(\kappa)$ for some constant $c$, where $\kappa$ is the amount of entropy in the left source.

\paragraph{Best of all worlds.}

Notice that the seeded non-malleable extractor, and the two-source extractors can be seen as special case of a two-source non-malleable extractor. With this view, the known constructions of negligible error (non-malleable) two-source extractors can be broadly classified in three categories:
\begin{itemize}
\item Constructions where one source has min-entropy rate about $1/2$, the other source can have small min-entropy rate, but the extractor doesn't guarantee non-malleability.
    \item Constructions where one source is uniform, and the other can have small min-entropy rate, and the extractor guarantees non-malleability when the uniform source is tampered.
    \item Constructions where both sources have entropy rate very close to $1$ and the extractor guarantees non-malleability against the tampering of both sources. 
\end{itemize}

The main focus of this work is the question whether we can have one construction that subsumes all the above constructions. 

\begin{question}
Is there an explicit construction of a two-source non-malleable extractor which requires two sources of length $n_1$ and $n_2$, and min-entropy requirement $c n_1$ (for some constant $c < 1$), and $\poly \log n_2$, respectively, that guarantees non-malleability against the tampering of both sources, and for which the error is negligible? In particular, can we obtain a construction with parameters suitable for application to privacy amplification with tamperable memory~\cite{AORSS20}?
\end{question}

In this work, we make progress towards answering this question.

\paragraph{Applications of two-source non-malleable extractors.} Two-source non-malleable  extractors have in the recent years attracted a lot of attention, and have very quickly become fundamental objects in cryptosystems involving communication channels that cannot be fully trusted.  As we discussed earlier, two-source non-malleable extractors have applications in the construction of non-malleable codes, and in constructing two-source extractors. The other primary applications of two-source non-malleable extractors include non-malleable secret sharing~\cite{GK18,ADNOPRS19}, non-malleable commitments~\cite{GPR16}, network extractors~\cite{GSZ21}, and privacy amplification~\cite{CKOS19,AORSS20}.  

In particular, in~\cite{AORSS20}, the authors introduce an extension of privacy amplification (PA) against active adversaries where, Eve as the active adversary is additionally allowed to \emph{fully corrupt} the internal memory of one of the honest parties, Alice and Bob, before the execution of the protocol.  Their construction required two-source non-malleable extractors with one source having a small entropy rate $\delta$ (where $\delta$ is a constant close to $0$). Since no prior construction of two-source non-malleable extractor satisfied these requirements, the authors constructed such extractors under computational assumptions and left the construction of the information-theoretic extractor with the desired parameters as an open problem. Our construction in this work resolves this open problem. We do not include here the details of the PA protocol due to space constraints. We refer the reader to~\cite{AORSS20} for the PA protocol.

\paragraph{Subsequent work.} Li, inspired by our work and that of~\cite{GSZ21}, in \cite{L23} gives a two-source non-malleable extractor construction
with $\frac{2}{3}$-rate entropy in one source and $\frac{\log(n)}{n}$-rate entropy in the other. Based on the proof sketch in~\cite{L23}, the key idea of the construction and proof seems similar, the fundamental difference being the use of an  correlation breaker with advice instead of a collision resistant extractor.

\paragraph{Our Contributions and Roadmap of the Paper.}

We build two-source non-malleable extractors, with one source having polylogarithmic min-entropy, and the other source having min-entropy rate $0.81$. We introduce collision-resistant extractors, and extend and improve efficiency of the privacy amplification protocol from~\cite{AORSS20}. The following is a roadmap of the paper.  
\begin{itemize}
    \item In Section~\ref{sec:overview}, we give an overview of our technical details. 
    \item In Section~\ref{sec:prelims}, we give mathematical preliminaries needed in the paper.
    
\item In Section~\ref{sec:generic-reduction}, we give a generic transformation that, takes in (1) a non-malleable two-source extractor which requires sources with high min-entropy, and (2) a two-source extractor which requires sources with smaller min-entropy and an additional collision-resistance property, and constructs a two-source non-malleable extractor with min-entropy requirement comparable to (but slightly worse) that of the two-source extractor used by the construction. 

\item In Section~\ref{sec:coll-res-generic}, we give a generic transformation that converts any seeded extractor (two-source extractor where one of the source is uniformly distributed)  to a collision-resistant seeded extractor with essentially the same parameters.

\item In Section~\ref{sec:coll-res-raz}, we show that the two-source extractor from~\cite{Raz05} is collision resistant.

\item In Section~\ref{sec:fullynmseeded}, we apply our generic transformation from Section~\ref{sec:coll-res-generic} to the seeded extractor from~\cite{RRV99} to obtain a collision-resistant seeded extractor. We then use the generic transformation from Section~\ref{sec:generic-reduction} along with the non-malleable extractor from~\cite{Li19} to obtain a two-source non-malleable extractor, where one of the source is uniform and the other has  min-entropy polylogarithmic in the length of the sources.

\item In Section~\ref{sec:nmraz}, we apply the generic transformation from Section~\ref{sec:generic-reduction} to the non-malleable extractor from Section~\ref{sec:fullynmseeded}, and the two-source extractor from~\cite{Raz05} to obtain a two-source non-malleable extractor where one source is required to have polylogarithmic min-entropy and the source is required to have min-entropy rate greater than $0.8$.

\item Applications:
\begin{itemize}
\item In Section~\ref{sec:ratehalfnm}, we use a generic transformation from~\cite{AKOOS21} to obtain a non-malleable two-source extractor where the length of the output is $1/2 - o(1)$ times the length of the input. Notice that via the probabilistic method, it can be shown that the output length of this construction is optimal. \footnote{The main drawback of this construction compared to the construction from Section~\ref{sec:nmraz} is that this is not a strong two-source non-malleable extractor, and hence cannot be used in most applications.}

\item In Section~\ref{App:PA}, we sketch the details of the privacy amplification protocol that uses our non-malleable two-source extractor.  We extend the protocol by \cite{AORSS20} to obtain a secret of optimal size while  maintaining security against a memory tampering adversary.
\end{itemize}
\end{itemize}

\section{Technical overview}
\label{sec:overview}

\subsection{Collision Resistant Extractors}
At the core of our non-malleable extractor compiler is a new object we call a \emph{collision resistant extractor}. An extractor is an object that takes as input two sources of randomness $X$ and $Y$ (in case of the seeded extractors $Y$ but uniform) and guarantees that, as long as $X$ and $Y$ are independent and have sufficient min-entropy, the output $\ext(X,Y)$ will be uniform (even given $Y$ \footnote{This property is often referred to as strong extraction \label{footstrong}}). A \emph{collision resistant extractor} $\innerExt$ has the added property that for all fixed-point-free functions $f$ (i.e. $f(x)\neq x$ for all $x$) the probability that $\innerExt(X,Y)=\innerExt(f(X),Y))$ is negligible \footnote{This notion might somewhat resemble various non-malleability notions, however in case of the non-malleability one would expect $\innerExt(f(X),Y))$ to be independent of $\innerExt(X,Y)$, here we only expect that those two outputs don't collide}. 

Readers might notice the resemblance to the collision resistant hashing families and the leftover hash lemma. The leftover hash lemma states that if the probability that $h(x_0,Y)=h(x_1,Y)$ is sufficiently small then $h(.,.)$ is an extractor. Obremski and Skorski (\cite{OS18}) showed that the inverse is almost true --- there exists a `core' of inputs on which every extractor has to fulfill the small collision probability property. This inverse leftover hash lemma is sadly not constructive and not efficient (the description of the core might be exponential), and thus we are unable to use it to obtain an efficient \emph{collision resistant extractor}.

We show that Raz's extractor (\cite{Raz05}) is a \emph{collision resistant extractor} with essentially the same parameters. We obtain this result by carefully modifying the original proof. The proof techniques are similar and we do not discuss the details in this section. 

We also show a generic transform that turns any seeded extractor (a two-source extractor where one source is uniform) into a \emph{collision resistant extractor} with a slight increase in the size of the seed.

\subsubsection{General Compiler for Seeded Extractors}

We first construct a collision-resistant extractor $h$ with a short output based on the Nisan-Widgerson generator \cite{NW94} or Trevisan's extractor \cite{RRV99}. Given the input $X$ and the seed $Z$, function $h$ will output $\hat{X}(Z_1)\concat\hat{X}(Z_2)\concat \cdots \concat\hat{X}(Z_t)$ where $\ecc$ is an error-correcting code of appropriate minimum distance, and $a \circ b$ denotes the concatenation of $a$ and $b$, $\hat{X} = \ecc(X)$, and $Z=Z_1\concat Z_2\concat \cdots \concat Z_t$, and $\hat{X}(Z_i)$ denotes $Z_i$-th bit of $\hat{X}$. Proof that this is an extractor follows directly from Nisan-Widgerson generator properties, while the collision resistance follows from the large distance of the error-correcting code.

We can now use any seeded extractor and the collision resistant extractor mentioned above to obtain a collision resistant seeded extractor with output size comparable to the seeded extractor.  Consider seeded extractors that take as input a random source $X$ and a short but uniform source $S$ and output $\ext(X,S)$ which is uniform (even given $S$ \footref{footstrong}).
Let us require on input a slightly longer uniform seed $S\concat Z$ (where $\concat$ denotes concatenation), and consider the following extractor: $\innerExt(X, S\concat Z)=\ext(X,S)\concat h(X,Z)$, where $h$ is either a collision resistant hash function or a collision resistant extractor. 

The proof follows quite easily. Function $h$ ensures that collisions indeed happen with negligible probability, the only thing left to show is that $\innerExt(X,S\concat Z)$ is uniform. 
First notice that by the definition the seeded extractor $\ext(X,S)$ is uniform, so we only have to show that $h(X,Z)$ is uniform even given $\ext(X,S)$. Observe that $Z$ is uniform and independent given $X,S$, so it suffices to show that $X$ has some remaining entropy given $\ext(X,S),S$, then $h(X,Z)$ will be uniform (either by leftover hash lemma, if $h$ is a collision resistant hash function, or by the definition of collision resistant extractor). This last step can be ensured simply by setting $\ext$ to  extract fewer bits than the entropy of $X$, thus a slight penalty in the parameters. 
Also notice that $h$ above can be a fairly bad extractor in terms of the rate or the output size and seed size. We can make the output and the seed of $h$ very small and thus the parameters of $\innerExt$ will be dominated by the parameters of $\ext$.


\subsection{Our Non-Malleable Extractor Compiler}

Our compiler takes as an input two objects, one is a collision resistant extractor (as discussed in the previous section), the other object is a strong two-source non-malleable extractor. A right-strong~\footnote{Notice that unlike many results in the literature, we need to distinguish between left strong and right strong for our extractor since the construction is inherently not symmetric.} non-malleable extractor gives the guarantee that $\ext(X,Y)$ is uniform even given $\ext(f(X),g(Y))$ and $Y$ (or $X$ in case of a left-strong non-malleable extractor) for any tampering functions $f,g$ where at least one of them are fixed-point-free. When we refer to a non-malleable extractor as strong without specifying if it's left-strong or right-strong we mean that the non-malleable extractor is both left-strong and right-strong. 
The construction is as follows: For a collision resistant extractor $\innerExt$, and a strong non-malleable extractor $\outerExt$ we consider following extractor:
\begin{equation}
    \twonmext(X, Y_\ell \concat Y_r) := \outerExt( Y_\ell \concat Y_r,  \innerExt(X, Y_\ell) ) \;.
\end{equation}
We will show that $\twonmext$ inherits the best of both worlds --- strong non-malleability of $\outerExt$ and the good entropy requirements of $\innerExt$.

There are two main issues to handle:
\paragraph{Issue of the independent tampering.}
Notice that the definition of the non-malleable extractor guarantees that $\ext(X,Y)$ is uniform given $\ext(X',Y')$ only if the sources are tampered independently (i.e. $X'$ is a function of only $X$, and $Y'$ is a function of only $Y$).

To leverage the non-malleability of $\outerExt$, we need to ensure that the tampering $X\rightarrow X'$ and $Y_\ell\concat Y_r \rightarrow Y'_\ell\concat Y'_r$ translates to the independent tampering of $Y_\ell\concat Y_r \rightarrow Y'_\ell \concat Y'_r$ and $\innerExt(X,Y_\ell) \rightarrow \innerExt(X',Y'_\ell)$. The problem is that both tamperings depend on $Y_\ell$. To alleviate this issue we will simply reveal $Y_\ell$ and $Y'_\ell$ (notice that $Y'_\ell$ can depend on $Y_r$ thus revealing $Y_\ell$ alone is not sufficient). Once $Y_\ell=y_\ell$ and $Y'_\ell=y'_\ell$ are revealed (and therefore fixed) the tampering $y_\ell\concat Y_r \rightarrow y'_\ell \concat Y'_r$ and $\innerExt(X,y_\ell) \rightarrow \innerExt(X',y'_\ell)$ becomes independent since right tampering depends only on $X$, which is independent of $Y_\ell\concat Y_r$ and remains independent of $Y_r$ even after we reveal $Y_\ell$ and $Y'_\ell$ (this extra information only lowers the entropy of $Y_r$).

\paragraph{Issue of the fixed points (or why we need collision resistance).} Non-malleable extractors guarantee that $\ext(X,Y)$ is uniform given $\ext(X',Y')$ if and only if $(X,Y) \neq (X',Y')$. 

The issue in our compiler is clear: If $Y_\ell \concat Y_r$ do not change, and $X$ is tampered to be $X'\neq X$ but $\innerExt(X',Y_\ell)=\innerExt(X,Y_\ell)$ then 
\begin{align*}
    &\twonmext(X, Y_\ell \concat Y_r) = \outerExt( Y_\ell \concat Y_r,  \innerExt(X, Y_\ell) ) \\
    &= \outerExt( Y_\ell \concat Y_r,  \innerExt(X', Y_\ell) ) =\twonmext(X', Y_\ell \concat Y_r)\;.
\end{align*}
To mitigate this problem, we require $\innerExt$ to be collision resistant, which means the probability that $\innerExt(X,Y_\ell)=\innerExt(X',Y_\ell)$ is negligible thereby resolving this issue. It is also possible to use $\innerExt$ without the collision resilience property, this gives a weaker notion of non-malleable extractor as was done in~\cite{GSZ21} .

\paragraph{Is $\twonmext$ strong?} Here we briefly argue that if $\outerExt$ is strong (i.e. both left and right strong) then $\twonmext$ will also be strong. To argue that compiled extractor is  left-strong, we notice that revealing $X$ on top of $Y_\ell$ and $Y'_\ell$ (which we had to reveal to maintain independence of tampering) translates to revealing $\innerExt(X,Y_\ell)$ which reveals right input of $\outerExt$ (revealing of $Y_\ell$ and $Y'_\ell$ is irrelevant since $Y_r$ maintains high enough entropy). As for the right-strongness, revealing $Y_r$ on top of $Y_\ell$ and $Y'_\ell$ translates to revealing of the left input of $\outerExt$, notice that $\innerExt(X,Y_\ell)$ remains uniform given $Y_\ell$ by the strong extraction property of $\innerExt$. 

$ $\\
For our construction, we will apply the compiler twice. First, we will use a collision resistant seeded extractor and the Li's extractor \cite{L19}. This gives us a strong non-malleable extractor $\fnmext$ for the first source with poly-logarithmic entropy, and the second source being uniform. We will refer to this object as a \emph{fully non-malleable seeded extractor}. We emphasize that this object is stronger than the seeded non-malleable extractor since it guarantees non-malleability for both sources. Then, we will then apply our compiler to Raz's extractor \cite{Raz05} and $\fnmext$ which will produce an extractor $\betterTwoNMExt$ that is a strong non-malleable extractor for the first source with poly-logarithmic entropy and the second source with entropy rate\footnote{Entropy rate is a ratio of min-entropy of the random variable to its length: $\frac{\minEnt{X}}{|X|}$}  $0.8$.

\subsubsection{Compiling Seeded Extractor with Li's Extractor}
In this section we will apply our compiler to the collision resistant seeded extractor $\crTreExt$ and strong non-malleable extractor $\liExt$ from $\cite{L19}$, yielding the following construction: 
\begin{equation}
    \fnmext(X,Y_\ell\concat Y_r)=\liExt(Y_\ell\concat Y_r, \crTreExt(X,Y_\ell)).
\end{equation}
The extractor $\liExt(\text{0.99}, \text{0.99})$ requires both sources to have a high entropy rate of $~99\%$\footnote{This is a simplification, formally speaking there exist a constant $\delta$ such that sources are required to have entropy rate above $1-\delta$. The reader may think of $\delta=0.01$.}, while the extractor $\crTreExt(\text{poly-log}, \text{uniform})$ requires first source to have poly-logarithmic entropy, and the second source to be uniform. 
Let us analyse the entropy requirements of the extractor $\fnmext$:
Since part of the construction is $\crTreExt(X,Y_\ell)$ we require $Y_\ell$ to be uniform, which means that whole $Y_\ell \concat Y_r$ has to be uniform. On the other hand $X$ has to only have a poly-logarithmic entropy. The output of $\crTreExt(X,Y_\ell)$ will be uniform which will fulfill the $0.99$ entropy rate requirement of $\liExt$. There is a small caveat: While $Y_\ell \concat Y_r$ is uniform one has to remember that we had to reveal $Y_\ell$ and $Y'_\ell$ to ensure independent tampering, therefore we only have to make sure that $Y_\ell$ is very short so $Y_\ell \concat Y_r$ will have over $0.99$ entropy rate even given $Y_\ell$ and $Y'_\ell$. This is possible since $\crTreExt$ requires only a very short seed length. Thus we get that $\fnmext(\text{poly-log},\text{uniform})$ requires first source to have poly-logarithmic entropy, while the second source is uniform, and non-malleability is guaranteed for both sources.

\subsubsection{Compiling Raz's Extractor with the Above}
Now we will compile Raz's extractor \cite{Raz05} with above obtained $\fnmext$. The result will be:
\begin{equation}
    \betterTwoNMExt(X,Y_\ell\concat Y_r)=\fnmext(Y_\ell\concat Y_r, \razExt(X,Y_\ell)).
\end{equation}
As we discussed above $\fnmext(\text{poly-log},\text{uniform})$ requires first source to have poly-logarithmic entropy, while the second source has to be uniform, $\razExt(\text{poly-log},0.5)$ requires first source to have poly-logarithmic entropy while the second source has to have over $0.5$ entropy rate. Therefore we require $Y_\ell$ to have an entropy rate above $0.5$ and it is sufficient if $X$ has poly-logarithmic entropy. As for requirements enforced by $\fnmext$, since the output of $\razExt$ will be uniform we only have check if $Y_\ell\concat Y_r$ has poly-logarithmic entropy given $Y_\ell$ and $Y'_\ell$. Given that $Y'_\ell$ can not lower the entropy of $Y_r$ by more than its size $|Y'_\ell|$ we have two equations:
\begin{align*}
    &\minEnt{Y_r} > |Y_\ell| \\
    &\minEnt{Y_\ell} > 0.5 |Y_\ell|
\end{align*}
which implies
\begin{align*}
    &\minEnt{Y_\ell\concat Y_r} > 2|Y_\ell| \\
    &\minEnt{Y_\ell\concat Y_r} > |Y_r| + 0.5 |Y_\ell|
\end{align*}
which asserts that $\frac{\minEnt{Y_\ell\concat Y_r}}{|Y_\ell\concat Y_r|} > 0.8$. Therefore $\betterTwoNMExt(\text{poly-log}, 0.8)$ requires first source to have poly-logarithmic entropy, while second source has to have entropy rate above $0.8$.

Finally notice that $\razExt$ has a relatively short output (shorter than both inputs) but that is not a problem since $\fnmext$ can have its first input much longer than the second input. We can adjust the output size of $\crTreExt$ to accommodate the input size requirements of $\liExt$ (this extractor requires both inputs to have the same length). We stress however that taking into consideration all inputs requirements both in terms of entropy and in terms of sizes is not trivial and our construction is tuned towards seeded-extractors and the Raz's extractor.

\section{Preliminaries}
\label{sec:prelims}

\subsection{Random Variables, Statistical Distance and Entropy}
For any set $S$, we denote by $\unif{S}$  the uniform distribution over the set $S$. For any positive integer $m$, we shorthand $\unif{\bit^m}$ by $\unif{m}$. For any random variable $X$, we denote the support of $X$ by $\supp{X}$. Also, for any random variable $X$ and event $E$, we denote by $X|_E$ the random variable $X'$ such that for all $x \in \supp{X}$, $\Pr[X'=x] = \Pr[X=x|E]$. 


\begin{definition}[Statistical Distance]
    Let $X, Y \in S$ be random variables. The \emph{statistical distance} between $X$ and $Y$ is defined by
    \begin{equation*}
\dist{X}{Y}:=        \frac{1}{2} \sum_{a \in S} \lvert \Pr[X = a] - \Pr[Y = a] \rvert
    \end{equation*}
    or equivalently,
    \begin{equation*}
    \dist{X}{Y}:=    \max_{A \subseteq S} \lvert \Pr[X \in S] - \Pr[Y \in S] \rvert.
    \end{equation*}
    
We shorthand the statement $\dist{X}{Y} \leq \eps$ by $X \closeTo{\eps} Y$ and we sometimes write this as $X$ is $\eps$-close to $Y$. 
\end{definition}

For any random variables $A, B, C$, and event $E$, we shorthand $\dist{A,C}{B,C}$ by $\distCond{A}{B}{C}$, and $\dist{A|_E}{B|_E}$ by $\distCond{A}{B}{E}$ i.e.,
    \begin{equation*}
        \distCond{A}{B}{C} = \dist{A, C}{B, C}\;,
    \end{equation*}
    and 
    \[
      \distCond{A}{B}{E} = \dist{A|_E}{B|_E}\;.
    \]

The following lemma is immediate from the definitions and triangle inequality.
\begin{lemma}
    Let $A, B, C$ be random variables such that $A, B \in S$ and $\supp{C} = T$ with $T = T_1 \cup T_2, T_1 \cap T_2 = \emptyset$. Then:

    \begin{enumerate}
        \item $\distCond{A}{B}{C} \le \sum_{c \in T} \Pr[C = c] \distCond{A}{B}{C = c}$ 
        \item $\distCond{A}{B}{C} \le \Pr[C \in T_1] \distCond{A}{B}{C \in T_1} + \Pr[C \in T_2] \distCond{A}{B}{C \in T_2}$
    \end{enumerate}
\end{lemma}

We will need the following standard lemmas.
\begin{lemma}[Lemma 10 of \cite{ADKO15} ]\label{lma:XORLemma}
    Let $\blocks{X}{m}$ be binary random variables and for any non-empty $\tau \subseteq [m]$, $\lvert \Pr[\bigoplus_{i \in \tau} X_i = 0] - \frac{1}{2} \rvert \leq \eps$, then $\dist{\blocks{X}{m}}{\unif{m}} \leq \eps \cdot 2^{\frac{m}{2}}$.
\end{lemma}

\begin{lemma}\label{lma:randDistinguisher}
    Let $X, Y$ be random variables. Further let $f_I$ be a family of functions $f$ indexed by set $I$ and let $S$ be a random variable supported on $I$ that is independent of both $X$ and $Y$. Then $f_S$ can be thought of as a randomised function such that $f_S(x) = f_s(x)$ with probability $\Pr[S = s]$.

    Then it holds that:
    \begin{equation*}
        \dist{f_S(X)}{f_S(Y)} \leq \dist{X}{Y}.
    \end{equation*}
\end{lemma}

\begin{lemma}[Lemma 4 of \cite{DDV10}, Lemma 9 of \cite{ADKO15}]\label{lma:reveal}
    Let $A, B$ be independent random variables and consider a sequence $\blocks{V}{i}$ of random variables, where for some function $\phi$, $V_i = \phi_i(C_i) = \phi(\blocks{V}{i-1}, C_i)$ with each $C_i \in \{A, B\}$. Then $A$ and $B$ are independent conditioned on $V_1, \ldots, V_i$. That is, $\mutualInfoCond{A}{B}{\blocks{V}{i}} = 0$.
\end{lemma}

\begin{definition}\label{defn:biasedAgainstTests}
    Call a sequence of variables $\blocks{Z}{N}$ $(k, \eps)$-biased against linear tests if for any non-empty $\tau \subseteq [N]$ such that $\lvert \tau \rvert \leq k$, $\lvert \Pr[\bigoplus_{i \in \tau} Z_i = 0] - \frac{1}{2} \rvert \leq \eps$.
\end{definition}

\begin{lemma}[Theorem 2 of \cite{AGHP02}]\label{lma:linearTests}
    Let $N = 2^t - 1$ and let $k$ be an odd integer. Then it is possible to construct $N$ random variables $Z_i$ with $i \in [N]$ which are $(k, \eps)$-biased against linear tests using a seed of size at most $2 \lceil \log(1/\eps) + \log \log N + \log k \rceil + 1$ bits.
\end{lemma}

\subsection{Min-entropy}

\begin{definition}[Min-entropy]
	Given a distribution $X$ over $\cX$, the \emph{min-entropy of $X$}, denoted by $\minEnt{X}$, is defined as
	\begin{equation*}
		\minEnt{X}=-\log\left(\max_{x\in\cX} \Pr[X=x]\right).
	\end{equation*}
\end{definition}

\begin{definition}[Average min-entropy]
	Given distributions $X$ and $Z$, the \emph{average min-entropy of $X$ given $Z$}, denoted by $\avgCondMinEnt{X}{Z}$, is defined as
	\begin{equation*}
		\avgCondMinEnt{X}{Z}=-\log\left(\mathbb{E}_{z\leftarrow Z}\left[\max_{x\in\cX} \Pr[X=x|Z=z]\right]\right).
	\end{equation*}
\end{definition}

\begin{lemma}[\cite{DORS08}]\label{lem:avgminH}
    Given arbitrary distributions $X$ and $Z$ such that $|\supp{Z}|\leq 2^\lambda$, we have
    \begin{equation*}
        \avgCondMinEnt{X}{Z}\geq \minEnt{X,Z}-\lambda\geq \minEnt{X}-\lambda\;.
    \end{equation*}
\end{lemma}

\begin{lemma}[\cite{MW97}]\label{lem:fixavgminH}
    For arbitrary distributions $X$ and $Z$, it holds that
    \begin{equation*}
        \Pr_{z\leftarrow Z}[\minEnt{X|Z=z}\geq \avgCondMinEnt{X}{Z}-s]\geq 1-2^{-s}.
    \end{equation*}
\end{lemma}

\begin{definition}[$(n, k)$-sources]
    We say that a random variable $X$ is an \emph{$(n, k)$-source} if $\supp{X} \subseteq \bit^n$ and $\minEnt{X} \ge k$. Additionally, we say that $X$ is a \emph{flat $(n, k)$-source} if for any $a \in \supp{X}$, $\Pr[X = a] = 2^{-k}$, i.e., $X$ is uniform over its support. 
\end{definition}
$X \givenBy (n, k)$ denotes the fact that $X$ is an $(n, k)$-source. Further, we call $X$ \emph{$(n, k)$}-flat if $X \givenBy (n, k)$ and is flat. We say that $X$ is $\eps$-close to a flat distribution if there exists a set $S$ such that $X \closeTo{\eps} \unif{S}$.

\begin{definition}[$\eps$-smooth min-entropy]
    A random variable $X$ is said to have $\eps$-smooth min-entropy at least $k$ if there exists $Y$ such that $\dist{X}{Y} \leq \eps$, and

    \begin{equation*}
     \minEnt{Y} \ge k\;.
    \end{equation*}
\end{definition}

\subsection{Extractors}
\begin{definition}[(Strong) Two-Source Extractor, Collision Resistance]\label{defn:twosourceExt}
    Call $E:\bit^{n_1} \times \bit^{n_2} \to \bit^m$ a \emph{two-source extractor} for input lengths $n_1, n_2$, min-entropy $k_1, k_2$, output length $m$, and error $\eps$ if for any two independent sources $X, Y$ with $X \givenBy (n_1, k_1)$, $Y \givenBy (n_2, k_2)$, the following holds:
    \begin{equation*}
        \dist{E(X, Y)}{\unif{m}} \leq \eps
    \end{equation*}
    If $n_2= k_2$, we call such an extractor \emph{seeded}. We use $\extractorParam{E}{(n_1, k_1)}{(n_2, k_2)}{m}{\eps}$ to denote the fact that $E$ is such an extractor.

    Additionally, we call the extractor $E$ \emph{right strong}, if:
    \begin{equation*}
        \distCond{E(X, Y)}{\unif{m}}{Y} \leq \eps \;,
    \end{equation*}
    and we call the extractor $E$ \emph{left strong}, if:
    \begin{equation*}
        \distCond{E(X, Y)}{\unif{m}}{X} \leq \eps \;.
    \end{equation*}
    
    We call an extractor $E$ \emph{strong} if it is both left strong and right strong.
    
The extractor is said to be $\collisionError$-collision resistant if $\Pr_{X, Y}[E(X, Y) = E(f(X), Y)] \leq \collisionError$ for all fixed-point-free functions $f$.
\end{definition}

\begin{definition}[Two Source Non-malleable Extractor]\label{defn:twosourceNMext}
    Call $\extractorParam{E}{(n_1, k_1)}{(n_2, k_2)}{m}{\eps}$ a \emph{two source non-malleable extractor} if additionally for any pair of functions $f : \bit^{n_1} \to \bit^{n_1}$, $g : \bit^{n_2} \to \bit^{n_2}$ such at least one of $f, g$ is fixed-point-free\footnote{A function $f$ is said to be fixed-point-free if for any $x$, $f(x) \neq x$}, the following holds:

    \begin{equation*}
        \distCond{E(X, Y)}{\unif{m}}{E(f(X), g(Y))} \leq \eps
    \end{equation*}

 Additionally, we call the extractor $E$ a \emph{right strong non-malleable two-source extractor} if:
    \begin{equation*}
        \distCond{E(X, Y)}{\unif{m}}{E(f(X), g(Y)),Y} \leq \eps \;,
    \end{equation*}
    and 
    we call the extractor $E$ a \emph{left strong non-malleable two-source extractor} if:
    \begin{equation*}
        \distCond{E(X, Y)}{\unif{m}}{E(f(X), g(Y)),X} \leq \eps \;,
    \end{equation*}
\end{definition}

\begin{definition}[(Fully) Non-malleable Seeded Extractor]
    Call $\extractorParam{E}{(n_1, k_1)}{(n_2, n_2)}{m}{\eps}$ a \emph{non-malleable seeded extractor} if additionally for some fixed-point-free function $g : \bit^{n_2} \to \bit^{n_2}$, the following holds:
    
    \begin{equation*}
        \distCond{E(X, Y)}{\unif{m}}{E(X, g(Y))} \leq \eps
    \end{equation*}

    A natural strengthening of a non-malleable seeded extractor is to consider a pair of tampering functions on both its inputs rather than on just the seed. Thus call a $E$ a \emph{fully non-malleable seeded extractor} if additionally for some pair of fixed-point-free functions $g : \bit^{n_2} \to \bit^{n_2}$, and $f : \bit^{n_1} \to \bit^{n_1}$, the following holds:
    
    \begin{equation*}
        \distCond{E(X, Y)}{\unif{m}}{E(f(X), g(Y))} \leq \eps
    \end{equation*}
\end{definition}

One useful thing to note is that the extractor remains non-malleable even if the functions $f, g$ are randomised with shared coins (independent of $X$ and $Y$).
\begin{lemma}\label{lma:randTamp}
    Let $E$ be a two source non-malleable extractor for $(n, k)$-sources $X, Y$ with output length $m$ and error $\eps$. Let $f_S, g_S$ random functions over the shared randomness of $S$ independent of $X$ and $Y$ such that for all $s \in \supp{S}$, at at least one of $f_s$ or $g_s$ is fixed-point-free. Then
    \begin{equation*}
        \distCond{E(X, Y)}{\unif{m}}{E(f_S(X), g_S(Y))} \leq \eps
    \end{equation*}
\end{lemma}
\begin{proof}
    Let $X \givenBy (n, k)$ and $Y \givenBy (n, k)$ be independent sources. Let $f_S, g_S : \bit^n \to \bit^n$ be fixed-point-free random functions over the randomness of $S$ which is independent of $X$ and $Y$.    
    \begin{align*}
        &\distCond{E(X, Y)}{\unif{m}}{E(f_S(X), g_S(Y))} \\&= \sum_{a, b} \lvert \Pr[E(X, Y) = a, E(f_S(X), g_S(Y)) = b] - \Pr[\unif{m}, E(f_S(X), g_S(Y)) = b] \rvert\\
        &= \sum_{a, b} \lvert \sum_s \Pr[S = s] \Pr[E(X, Y) = a, E(f_S(X), g_S(Y)) = b | S = s]\\
        & \hspace*{2em} - \sum_s \Pr[S = s] \Pr[\unif{m}, E(f_S(X), g_S(Y)) = b | S = s]\rvert\\
        &= \sum_{a, b} \sum_s \Pr[S = s] \lvert \Pr[E(X, Y) = a, E(f_S(X), g_S(Y)) = b | S = s] - \Pr[\unif{m}, E(f_S(X), g_S(Y)) = b | S = s]\rvert\\
        &= \sum_{a, b} \sum_s \Pr[S = s] \lvert \Pr[E(X, Y) = a, E(f_s(X), g_s(Y)) = b] - \Pr[\unif{m}, E(f_s(X), g_s(Y)) = b]\rvert\\
        &= \sum_s \Pr[S = s] \sum_{a, b} \lvert \Pr[E(X, Y) = a, E(f_s(X), g_s(Y)) = b] - \Pr[\unif{m}, E(f_s(X), g_s(Y)) = b]\rvert\\
        &= \sum_s \Pr[S = s] \distCond{E(X, Y)}{\unif{m}}{E(f_s(X), g_s(Y))}\\
        &\leq \sum_s \Pr[S = s] \eps = \eps
    \end{align*}
    Note that $S$ is independent of $X$ and $Y$ and thus $E(X, Y)$ is independent of $S$. Now for a fixed $s$, $f_s$ and $g_s$ are fixed functions. So the last inequality follows as $E$ is a two source non-malleable extractor.
\end{proof}

\begin{lemma}\label{lma:avgCaseMinExtraction}
    If $\extractorParam{\ext}{(n, k)}{(d, d)}{m}{\eps}$ is a strong seeded extractor, then for any $X, W$ such that $\supp{X} \subseteq \bit^n$ and $\avgCondMinEnt{X}{W} \geq k + \log(1/\eta)$ with $\eta > 0$, it holds that:

    \begin{equation*}
        \distCond{\ext(X, \unif{d})}{\unif{m}}{\unif{d}, W} \leq \eps + \eta 
    \end{equation*}
\end{lemma}

\begin{proof}
    Let $ext$, $X$ and $W$ be defined as above. Then, given that $\avgCondMinEnt{X}{W} \geq k + \log(1/\eta)$, 
    it follows from Markov's inequality that there exists a ``bad'' set $B$ such that $\Pr[W \in B] \leq \eta$, and 
    for all $w \notin B$, $\condMinEnt{X}{W = w} \geq k$.
    Then,
    \begin{align*}
        \distCond{\ext(X, \unif{d})}{\unif{m}}{\unif{d}, W}
        &\leq \distCond{\ext(X, \unif{d})}{\unif{m}}{\unif{d}, W \in B} \Pr[W \in B]\\
        &+ \distCond{\ext(X, \unif{d})}{\unif{m}}{\unif{d}, W \notin B} \Pr[W \notin B]\\
        &\leq 1 \cdot \Pr[W \in B] + \distCond{\ext(X, \unif{d})}{\unif{m}}{\unif{d}, W \notin B}\\
        &= \Pr[W \in B] + \sum_{w \notin B} \distCond{\ext(X, \unif{d})}{\unif{m}}{\unif{d}, W = w}\\
        &\leq \eta + \eps\;.
    \end{align*}
\end{proof}

We will need the following constructions of extractors. 
\begin{lemma}[Theorem 6.9 of \cite{L19}]\label{lma:li}
    There exists a constant $0 < \gamma < 1$ and an explicit two-source non-malleable extractor $\extractorParam{\liExt}{(n, (1-\gamma)n)}{(n, (1-\gamma)n)}{\Omega(n)}{\liError}$ such that $\liError = 2^{-\Omega(n \frac{\log \log n}{\log n})}$.
\end{lemma}

\begin{lemma}[Theorem 2 of \cite{RRV99}]\label{lma:tre}
    For every $n, k$ there exists an explicit strong seeded extractor $\extractorParam{\treExt}{(n, k)}{(d, d)}{\Omega(k)}{\eps}$ such that $d = O(\log^2(n)\log(1/\eps))$.
\end{lemma}

\begin{lemma}[Theorem 1 of \cite{Raz05}]\label{lma:raz}
    For any $n_1, n_2, k_1, k_2, m$ and any $0 < \delta < \frac{1}{2}$ such that:
    \begin{enumerate}
        \item $k_1 \geq 5\log(n_2 - k_2)$
        \item $n_2 \geq 6 \log n_2 + 2\log n_1$,
        \item $k_2 \geq (\frac{1}{2} + \delta) \cdot n_2 + 3\log n_2 + \log n_1$,
        \item $m = \Omega(\min \{ n_2, k_1 \})$,
    \end{enumerate}
    there exists a strong two-source extractor $\extractorParam{\razExt}{(n_1, k_1)}{(n_2, k_2)}{m}{\eps}$, such that $\eps = 2^{-\frac{3m}{2}}$.
\end{lemma}



\subsubsection{Rejection Sampling for Extractors}
In this section we present two lemmas that use rejection sampling to lower the entropy requirement for strong two-source extractors and their collision resistance.

We first define a sampling algorithm $\samp$ that given
a flat distribution $Y' \givenBy (n, k)$, tries to approximate some distribution $Y \givenBy (n, k - \delta)$ (with $supp(Y) \subseteq supp(Y')$). Letting $d = \max_{y \in supp(Y)} \left \{ \frac{\Pr[Y = y]}{\Pr[Y' = y]} \right \}$:

\begin{equation*}
    \samp(y) = \begin{cases}
        y, & w.p.\ \frac{\Pr[Y = y]}{d \cdot \Pr[Y' = y]}\\
        \bot, &else 
    \end{cases}
\end{equation*}

\begin{lemma}\label{lma:rejectionSample}
    The probability $\samp(Y') = y$ is $\frac{\Pr[Y = y]}{d}$ and furthermore, the probability that $\samp(Y') \neq \bot$ is $\frac{1}{d}$. Consequently, the distribution $\samp(Y')$ conditioned on the event that $\samp(Y') \neq \bot$ is identical to $Y$.
\end{lemma}

\begin{proof}
    Letting $\samp$ and $d$ be defined as above, then:
    \begin{equation*}
        \Pr[\samp(Y') = y] = \frac{1}{d} \frac{\Pr[Y = y]}{\Pr[Y' = y]} \cdot \Pr[Y' = y] = \frac{\Pr[Y = y]}{d}
    \end{equation*}

    Then it follows that:
    \begin{equation*}
        \Pr[\samp(Y') \neq \bot] = \sum_{y} \Pr[\samp(Y) = y] = \sum_{y} \frac{\Pr[Y = y]}{d} = \frac{1}{d}
    \end{equation*}

    Thus, conditioned on the event that $\samp(Y') \neq \bot$, $\samp(Y')$ is the distribution $Y$.
    \begin{align*}
        \Pr[\samp(Y') = y | \samp(Y') \neq \bot] &= \frac{\Pr[\samp(Y') \neq \bot |\samp(Y') = y ] \Pr[\samp(Y') = y]}{\Pr[\samp(Y') \neq \bot]} \\&= \Pr[Y = y]
    \end{align*}
\end{proof}

\paragraph{Lowering the Entropy Requirement for Strong Two-Source Extractors.}
\begin{lemma}\label{lma:extReduction}
    Let $\extractorParam{\ext}{(n_1, k_1)}{(n_2, k_2)}{m}{\eps}$ be a strong two-source extractor using input distributions $X$ and $Y'$. Then letting $Y \givenBy (n_2, k_2 - \delta)$:
    
    \begin{equation*}
        \distCond{\ext(X, Y)}{\unif{m}}{Y} \leq 2^\delta \eps
    \end{equation*}
\end{lemma}

\begin{proof}
    Assume by contradiction that there exists a distribution
    $Y \givenBy (n, k -\delta)$ for which $\distCond{\ext(X, Y)}{\unif{m}}{Y} > 2^\delta \eps$, i.e. there exists a distinguisher $A : \bit^m \to \bit$ such that $\lvert \Pr[A(\ext(X, Y), Y) = 1] - \Pr[A(\unif{m}, Y) = 1] \rvert > 2^\delta \eps$. We want to use this fact to 
    create a distinguisher $D$ that distinguishes $\ext(X, Y')$
    from $\unif{m}$ \emph{for some distribution $Y'\givenBy (n, k)$}. Note that 
    $Y$ can be expressed as a convex combination
    of $(n, k - \delta)$ flat distributions,  i.e. $Y = \sum_i \alpha_i Y_i$.
    We define $Y'$ in the following way:
    $Y'$ is a convex combination of flat 
    distributions $Y'_i$ where each $Y'_i$
    is some $(n, k)$ flat distribution such that
    $supp(Y_i) \subseteq supp(Y'_i)$. We first note
    that for all $y \in supp(Y)$:
    \begin{equation*}
        \frac{Pr[Y = y]}{Pr[Y' = y]} = \frac{\sum_i \alpha_i Pr[Y_i = y] }{ \sum_i \alpha_i Pr[Y'_i = y] }
        \leq \frac{2^{-k + \delta}}{2^{-k}} \leq 2^\delta
    \end{equation*}
    Furthermore, note that $Y'$ has min-entropy $k$. To see this, note that for any $y \in supp(Y)$:
    \begin{equation*}
        \Pr[Y' = y] = \sum_i \alpha_i \Pr[Y'_i = y]
        \leq \sum_i \alpha_i 2^{-k}
        \leq 2^{-k}
    \end{equation*}
    
    Let 
    $\samp$ be a rejection sampler that on input distribution $Y$, samples for $Y'$.
    Now, $D$ is defined as follows:
    \begin{equation*}
        D(Z, Y') = \begin{cases}
                        A(Z, Y') &,\ \text{if }\samp(Y') \neq \bot\\
                        1 &,\ w.p.\ \frac{1}{2},\ \text{if }\samp(Y') = \bot\\
                        0 &,\ else
                    \end{cases}
    \end{equation*}
    Note that by Lemma \ref{lma:rejectionSample}, $\Pr[\samp(Y') \neq \bot] \geq \frac{1}{2^\delta}$ and $\samp(Y')$ is identical to $Y$ conditioned on the event that $\samp(Y') = \bot$. Then the advantage that $D$ distinguishes between $\ext(X, Y')$ and $\unif{m}$ given $Y'$ is given as:
    \begin{align*}
        &\lvert \Pr[D(\ext(X, Y'), Y') = 1] - \Pr[D(\unif{m}, Y') = 1] \rvert\\
        &\geq \Pr[\samp(Y) \neq \bot] \lvert \Pr[A(\ext(X, Y), Y) = 1]
         - \Pr[A(\unif{m}, Y) = 1] \rvert\\
        &>\frac{1}{2^\delta} 2^\delta \eps = \eps
    \end{align*}
    Which in turn implies that $\distCond{\ext(X, Y')}{\unif{m}}{Y'} > \eps$, which implies the desired contradiction.
\end{proof}




\paragraph{Lowering the Entropy Requirement for Collision Resistance in Extractors.}
\begin{lemma}\label{lma:colReduction}
    Let $\extractorParam{\ext}{(n_1, k_1)}{(n_2, k_2)}{m}{\eps}$ be a strong two-source extractor using input distributions $X$ and $Y'$ that has collision probability $\collisionError$. Then letting $Y \givenBy (n_2, k_2 - \delta)$ and $f$ be any fixed-point-free function:
    
    \begin{equation*}
        \Pr[\ext(X, Y) = \ext(f(X), Y)] \leq 2^\delta \collisionError
    \end{equation*}
\end{lemma}
\begin{proof}
    For the sake of contradiction, $Y$ be any $(n, k - \delta)$ distribution for which
    the collision probability is at least 
    $2^{\delta} \cdot \collisionError$.
    
    Note that 
    $Y$ can be expressed as a convex combination
    of $(n, k - \delta)$ flat distributions, i.e. $Y = \sum_i \alpha_i Y_i$.
    We define $Y'$ in the following way:
    $Y'$ is a convex combination of flat 
    distributions $Y'_i$ where each $Y'_i$
    is some $(n, k)$ flat distribution such that
    $supp(Y_i) \subseteq supp(Y'_i)$. We first note
    that for all $y \in supp(Y)$:
    \begin{align*}
        \frac{Pr[Y = y]}{Pr[Y' = y]} &= \frac{\sum_i \alpha_i Pr[Y_i = y] }{ \sum_i \alpha_i Pr[Y'_i = y] } \\
        &\leq \frac{2^{-k + \delta}}{2^{-k}} \leq 2^\delta
    \end{align*}
    Furthermore, note that $Y'$ has min-entropy $k$. To see this, note that for any $y \in supp(Y)$:
    \begin{equation*}
        \Pr[Y' = y] = \sum_i \alpha_i \Pr[Y'_i = y]
        \leq \sum_i \alpha_i 2^{-k}
        \leq 2^{-k}
    \end{equation*}

    Let $\samp$ be a rejection sampler that on input distribution
    $Y'$, samples for $Y$. By the collision resilience property of $\ext$, it follows that:
    \begin{align*}
        \collisionError &\geq \Pr[ \ext(X, Y') = \ext(f(X), Y') ]\\
                        &\geq \Pr[ \ext(X, Y) = \ext(f(X), Y) | \samp(Y) \neq \bot ] \Pr[ \samp(Y) \neq \bot ]\\
                        &= \Pr[ \ext(X, Y) = \ext(f(X), Y)] 2^{-\delta}
    \end{align*}
\end{proof}

\section{A Generic Construction of a Two-Source Non-Malleable Extractor}
\label{sec:generic-reduction}
In this section we present a generic construction that transforms a non-malleable two-source extractor $\mathbf{E}$ into another non-malleable two-source extractor with a much smaller entropy rate requirement via a two-source extractor.

\begin{theorem}
\label{thm:main_reduction}
 For any integers $n_1, n_2, n_3, n_4, k_1, k_2, k_3, k_4, m$ and $\outerError, \innerError, \collisionError > 0$, $n_4 < n_1$, given an efficient construction of
 \begin{itemize}
    \item a strong non-malleable extractor  $\extractorParam{\outerExt}{(n_1, k_1)}{(n_2, k_2)}{m}{\outerError}$,
    \item a right strong two-source extractor $\extractorParam{\innerExt}{(n_3, k_3)}{(n_4, k_4)}{n_2}{\innerError}$ that is $\collisionError$-collision resistant,
\end{itemize}
then for any integers $k_1^*, k_2^*$, $\eps, \tau > 0$ that satisfy the following conditions, there is an efficient construction of a left and right strong non-malleable two-source extractor $\extractorParam{\twonmext}{(n_3, k_1^*)}{(n_1, k_2^*)}{m}{\eps}$. \[k_1^* \ge k_3 \;,\]
\[k_2^* \ge \log 1/\tau + \max\left(k_4 + (n_1 - n_4), k_1 + 2 n_4 \right)   \;,\]
and 
\[
\eps \le 3\tau + 3 \outerError + 2 \innerError + 2 \sqrt{\collisionError}\;.
\]
\end{theorem}
\begin{proof}
Our construction is as follows: Given inputs $x \in \bit^{n_3}$ and $y = y_\ell \concat y_r$, where $y_\ell \in \bit^{n_4}$, and $y_r \in \bit^{n_1-n_4}$ our extractor is defined as:
\begin{equation}
    \twonmext(x, y) := \outerExt( y_\ell \concat y_r,  \innerExt(x, y_\ell) ) \;.
\end{equation}

Let $f:\bit^{n_3} \to \bit^{n_3}$ and $g:\bit^{n_1} \to \bit^{n_1}$. For any $y \in \bit^{n_1}$, by $g(y)_\ell$ we denote the $n_4$ bit prefix of $g(y)$. We assume that $f$ does not have any fixed points. The proof for the case when $g$ not having any fixed points is similar (in fact, simpler) as we explain later. 

\paragraph{Right strongness.} We first prove that our non-malleable extractor is right strong. 

\begin{claim}\label{clm:CloseToUnif}
 Let $\tY$ be a random variable with min-entropy $k_2^* - \log 1/\tau$ and is independent of $X$.   Consider the randomized function $\modifiedTampering$ that given $a, b, c$, samples $\innerExt(f(X), c)$ conditioned on $\innerExt(X, b) = a$, i.e., \[\modifiedTampering : a, b, c \mapsto \innerExt(f(X), c)_{| \innerExt(X, b) = a}\;.\] Then:

     \begin{equation}
    \distCond{
        \begin{array}{c}
            \innerExt(X, \tY_\ell)\\\outerExt( \tY_\ell \concat \tY_r,  \innerExt(X, \tY_\ell) )\\
            \outerExt( g(\tY),  \innerExt(f(X), g(\tY)_\ell) )            
        \end{array}
    }{
        \begin{array}{c}
            \unif{d}\\\outerExt( \tY_\ell \concat \tY_r,  \unif{d} )\\
            \outerExt( g(\tY),  \modifiedTampering(\unif{\innerOutputLength}, \tY_\ell, g(\tY)_\ell) )
        \end{array}
    }{
        \begin{array}{c}
            \tY_r\\ \tY_\ell\\
            g(\tY)_\ell
        \end{array}
    } \leq \innerError \;.
    \end{equation}
\end{claim}
\begin{proof}
 We have that $\minEnt{X} \geq k_1^* \ge k_3$ and $\minEnt{\tY_\ell} \geq k_2^* - \log 1/\tau - |\tY_r| = k_2^* - \log 1/\tau -  (n_1 - n_4) \ge k_4$, and $X, \tY_\ell$ are independently distributed. It follows that $\distCond{\innerExt(X, \tY_\ell)}{\unif{\innerOutputLength}}{\tY_\ell} \leq \innerError$. Then,  Lemma~\ref{lma:randDistinguisher} implies that
    \[        \distCond{\innerExt(X, \tY_\ell)}{\unif{\innerOutputLength}}{\tY_\ell, \tY_r, g(\tY)_\ell} \le \innerError \;.
    \]

Observing that since $\tY_r$ is independent of $\innerExt(f(X), g(\tY)_\ell), \innerExt(X, \tY_\ell)$ given $\tY_\ell, g(\tY)_\ell$, we have that the tuple $\innerExt(X, \tY_\ell), \tY_\ell, \tY_r,  \modifiedTampering(\innerExt(X, \tY_\ell), \tY_\ell, g(\tY)_\ell)$ is identically distributed as $\innerExt(X, \tY_\ell), \tY_\ell, \tY_r,  \innerExt(f(X), g(\tY)_\ell)$.   Again applying Lemma \ref{lma:randDistinguisher}, we get the desired statement.
 \end{proof}
    
    Now, let $\cY_0$ be the set of $y$ such that $g(y)_\ell = y_\ell$, and $\cY_1$ be the set of all $y$ such that $g(y)_\ell \neq y_\ell$ (in other words, $\cY_0$ contains all the fixed-points of $g$, and $\cY_1$ is the complement set). Also, let $\cY_{0,0}$ be the set of all $y \in \cY_0$ such that $\Pr[C(X,y_\ell) = C(f(X), y_\ell)] \le \sqrt{\collisionError}$, and $\cY_{0,1} = \cY_0 \setminus \cY_{0,0}$.  
    
    \begin{claim}\label{cl:Y1}
        If $\Pr[Y \in \cY_1] \ge \tau$, then 
        \begin{equation}
    \distCond{
        \begin{array}{c}
            \outerExt( \tY_\ell \concat \tY_r,  \innerExt(X, \tY_\ell) )\\
            \outerExt( g(\tY),  \innerExt(f(X), g(\tY)_\ell) ) 
        \end{array}
    }{
        \begin{array}{c}
           U_m\\
            \outerExt( g(\tY),  \innerExt(f(X), g(\tY)_\ell) )
        \end{array}
    }{  
        \begin{array}{c}
        \tY_\ell\\ \tY_r
        \end{array}
    } \leq \innerError + \outerError  \;,
    \end{equation}
    where $\tY = Y|_{Y \in \cY_1}$. 
    \end{claim}
    \begin{proof}
        Notice that conditioned on $Y$ being in $\cY_1$, $g$ does not have a fixed point. Thus, since $U_{n_2}$ is independent of $\tY_r$ given $\tY_\ell, g(\tY)_\ell$, and $H_\infty(U_{n_2}) = n_2 \ge k_2$, $H_\infty(\tY_r|\tY_\ell, g(\tY_r)) \ge k_2^* - \log 1/\tau - 2n_4 \ge k_1$, by the definition of a strong non-malleable extractor, we have that
        
    \[
    \distCond{
        \begin{array}{c}
            \outerExt( \tY_\ell \concat \tY_r,  U_{n_2} )\\
            \outerExt( g(\tY),  T_{f,g}(U_{n_2}, \tY_\ell, g(\tY)_\ell) )
        \end{array}
    }{
        \begin{array}{c}
           U_m\\
 \outerExt( g(\tY),  T_{f,g}(U_{n_2}, \tY_\ell, g(\tY)_\ell) )
        \end{array}
    }{
        \begin{array}{c}
        \tY_\ell\\ \tY_r
        \end{array}
    } \le \outerError \;. 
    \]
    Furthermore, from Claim~\ref{clm:CloseToUnif} and Lemma~\ref{lma:randDistinguisher}, we get that
    \[
    \distCond{
        \begin{array}{c}
            \outerExt( \tY_\ell \concat \tY_r,  U_{n_2} )\\
            \outerExt( g(\tY),  T_{f,g}(U_{n_2}, \tY_\ell, g(\tY)_\ell) )
        \end{array}
    }{
        \begin{array}{c}
        \outerExt( \tY_\ell \concat \tY_r,  \innerExt(X, \tY_\ell ))\\
            \outerExt( g(\tY),  \innerExt(f(X), g(\tY)_\ell )
        \end{array}
    }{
        \begin{array}{c}
        \tY_\ell\\ \tY_r
        \end{array}
    } \le \innerError \;. 
    \]
    The desired statement follows from triangle inequality.
    \end{proof}
    
    Similarly, we prove the following claim.
    
  \begin{claim}
        If $\Pr[Y \in \cY_{0,0}] \ge \tau$, then 
        \begin{equation}
    \distCond{
        \begin{array}{c}
            \outerExt( \tY_\ell \concat \tY_r,  \innerExt(X, \tY_\ell) )\\
            \outerExt( g(\tY),  \innerExt(f(X), g(\tY)_\ell) ), 
        \end{array}
    }{
        \begin{array}{c}
           U_m\\
            \outerExt( g(\tY),  \innerExt(f(X), g(\tY)_\ell) )
        \end{array}
    }{
        \begin{array}{c}
        \tY_\ell\\ \tY_r
        \end{array}
    } \leq \outerError + 2\innerError + \sqrt{\collisionError} \;,
    \end{equation}
    where $\tY = Y|_{Y \in \cY_{0,0}}$. 
    \end{claim}
    \begin{proof}
        Notice that  the probability that $\innerExt(X, \tY) = \innerExt(f(X), g(\tY)_\ell)$ is at most $\sqrt{\collisionError}$. Thus, by Claim~\ref{clm:CloseToUnif}, the probability that $U_{n_2} = T_{f,g}(U_{n_2}, \tY_\ell, g(\tY)_\ell)$ is at most $\sqrt{\collisionError} + \innerError$.   Also, since $U_{n_2}$ is independent of $\tY_r$ given $\tY_\ell, g(\tY)_\ell$, and $H_\infty(U_{n_2}) = n_2 \ge k_2$, $H_\infty(\tY_r|\tY_\ell, g(\tY_r)) \ge k^* - \log 1/\tau - 2n_4 \ge k_2$, by the definition of a strong non-malleable extractor, we have that
        
    \begin{equation*}
    \distCond{
        \begin{array}{c}
            \outerExt( \tY_\ell \concat \tY_r,  U_{n_2} )\\
            \outerExt( g(\tY),  T_{f,g}(U_{n_2}, \tY_\ell, g(\tY)_\ell) )
        \end{array}
    }{
        \begin{array}{c}
           U_m\\
 \outerExt( g(\tY),  T_{f,g}(U_{n_2}, \tY_\ell, g(\tY)_\ell) )
        \end{array}
    }{\begin{array}{c}
        \tY_\ell\\ \tY_r
        \end{array}} \le {\begin{array}{c}
             \outerError + \\
             \innerError + \\
              \sqrt{\collisionError}
        \end{array}}   \;. 
    \end{equation*}
    
    Furthermore, from Claim~\ref{clm:CloseToUnif} and Lemma~\ref{lma:randDistinguisher}, we get that
    \[
    \distCond{
        \begin{array}{c}
            \outerExt( \tY_\ell \concat \tY_r,  U_{n_2} )\\
            \outerExt( g(\tY),  T_{f,g}(U_{n_2}, \tY_\ell, g(\tY)_\ell) )
        \end{array}
    }{
        \begin{array}{c}
        \outerExt( \tY_\ell \concat \tY_r,  \innerExt(X, \tY_\ell ))\\
            \outerExt( g(\tY),  \innerExt(f(X), g(\tY)_\ell )
        \end{array}
    }{\begin{array}{c}
        \tY_\ell\\ \tY_r
        \end{array}} \le \innerError \;. 
    \]
    The desired statement follows from triangle inequality.
    \end{proof}
    We now show that $Y \in \cY_{0,1}$ with small probability.
    
    \begin{claim}
    \label{claim:Y01}
       \[ \Pr[Y \in \cY_{0,1}] \le \tau + \sqrt{\collisionError}\;.\]
    \end{claim}
    \begin{proof}
        If $\Pr[Y \in \cY_0] < \tau$, then the statement trivially holds. So, we assume $\Pr[Y \in \cY_0] \ge \tau$. Let $\tY = Y|_{Y \in \cY_0}$  Then $H_\infty(\tY) \ge k_2^* - \log 1/\tau - (n_1 - n_4) \ge k_4$. Since $\mathbf{C}$ is collision-resistant, we have that 
        \begin{align*}
            \collisionError &\ge \Pr[\innerExt(X, \tY_\ell) = \innerExt(f(X), g(\tY)_\ell)] 
             \Pr[\tY \in \cY_{0,1}] \cdot \sqrt{\collisionError} \\
             &\ge \Pr[Y \in \cY_{0,1}] \cdot \sqrt{\collisionError} \;.
        \end{align*}
    \end{proof}
    We now conclude the proof of right strongness of our non-malleable extractor as follows. We shorthand         $\twonmext( X, Y), Y, \twonmext( f(X),  g(Y))$ by $\phi(X,Y)$, and $U_m, Y, \twonmext( f(X),  g(Y))$ by $\psi(X,Y)$.
    \begin{align*}
        \dist{
            \phi( X, Y)
    }{
        \psi(X,Y)
    } &\le \Pr[Y \in \cY_{0,1}]+ \Pr[Y \in \cY_1] \cdot \dist{
            \phi( X, Y)|_{Y \in \cY_1}
    }{
        \psi(X,Y)|_{Y \in \cY_1}
    } \\
    &\;\;\;\;\;\;\;\;\; + \Pr[Y \in \cY_{0,0}] \cdot \dist{
            \phi( X, Y)|_{Y \in \cY_{0,0}}
    }{
        \psi(X,Y)|_{Y \in \cY_{0,0}}
    } \\
    &\le (\tau + \outerError + \innerError) + (\tau+2\outerError + \innerError+\sqrt{\collisionError}) + (\tau+\sqrt{\collisionError}) \\
    &= 3\tau + 3 \outerError + 2 \innerError + 2 \sqrt{\collisionError} \;. 
    \end{align*}
    
    Note that we assumed that $f$ does not have fixed points. On the other hand, if $g$ does not have fixed points then a simpler proof works that does not need to partition the domain into $\cY_{0,0}, \cY_{0,1}, \cY_1$. Since the first source for the non-malleable extractor $\outerExt$, we can conclude the statement similar to Claim~\ref{cl:Y1} with $Y$ instead of $\tY$. 
    
    \paragraph{Left strongness.} 
 The proof of left strongness is nearly the same (the statistical distance statements include $X$ instead of $Y_r$), but we include it here for completeness. 

\begin{claim}\label{clm:CloseToUnif_2}
 Let $\tY$ be a random variable with min-entropy $k^* - \log 1/\tau$ and is independent of $X$.   Consider the randomized function $S$ that given $a, b$, samples $X$ conditioned on $\innerExt(X, b) = a$, i.e., \[S : a, b \mapsto X|_{ \innerExt(X, b) = a}\;.\] Then:

     \begin{equation}
    \distCond{
        \begin{array}{c}
            \innerExt(X, \tY_\ell) \\ X
        \end{array}
    }{
            \begin{array}{c}
            \unif{d},\\S(U_{n_2}, \tY_\ell)
             \end{array}
    }{
            \begin{array}{c}
        \tY_\ell\\ \tY_r
        \end{array}
    } \leq \innerError \;.
    \end{equation}
\end{claim}
\begin{proof}
 We have that $\minEnt{X} \geq k_1^* \ge k_3$ and $\minEnt{\tY_\ell} \geq k_2^* - \log 1/\tau - |\tY_r| = k_2^* - \log 1/\tau -  (n_1 - n_4) \ge k_4$, and $X, \tY_\ell$ are independently distributed. It follows that $\distCond{\innerExt(X, \tY_\ell)}{\unif{\innerOutputLength}}{\tY_\ell} \leq \innerError$. Then, using Lemma~\ref{lma:randDistinguisher} and 
observing that since $\tY_r$ is independent of $X$ given $\tY_\ell$, we have that  $\innerExt(X, \tY_\ell), \tY_\ell, \tY_r,  S(\innerExt(X, \tY_\ell), \tY_\ell)$ is identically distributed as $\innerExt(X, \tY_\ell), \tY_\ell, \tY_r, X$, we get the desired statement.
 \end{proof}
    
    Now, let $\cY_0$ be the set of $y$ such that $g(y)_\ell = y_\ell$, and $\cY_1$ be the set of all $y$ such that $g(y)_\ell \neq y_\ell$. Also, let $\cY_{0,0}$ be the set of all $y \in \cY_0$ such that $\Pr[C(X,y_\ell) = C(f(X), y_\ell)] \le \sqrt{\collisionError}$, and $\cY_{0,1} = \cY_0 \setminus \cY_{0,0}$.  
    
    \begin{claim}\label{cl:Y1_2}
        If $\Pr[Y \in \cY_1] \ge \tau$, then 
        \begin{equation}
    \distCond{
        \begin{array}{c}
            \outerExt( \tY_\ell \concat \tY_r,  \innerExt(X, \tY_\ell) )
        \end{array}
    }{
        \begin{array}{c}
           U_m
        \end{array}
    }{
        \begin{array}{c}
        
        \tY_\ell\\ g(\tY)_\ell\\ X\\ \outerExt( g(\tY),  \innerExt(f(X), g(\tY)_\ell) )
        \end{array}
    } \leq 2\innerError + \outerError  \;,
    \end{equation}
    where $\tY = Y|_{Y \in \cY_1}$. 
    \end{claim}
    \begin{proof}
        Notice that conditioned on $Y$ being in $\cY_1$, $g$ does not have a fixed point. Thus, since $U_{n_2}$ is independent of $\tY_r$ given $\tY_\ell, g(\tY)_\ell$, and $H_\infty(U_{n_2}) = n_2 \ge k_2$, $H_\infty(\tY_r|\tY_\ell, g(\tY_r)) \ge k^* - \log 1/\tau - 2n_4 \ge k_1$, by the definition of a strong non-malleable extractor, we have that
        
    \[
    \distCond{
        \begin{array}{c}
            \outerExt( \tY_\ell \concat \tY_r,  U_{n_2} )
        \end{array}
    }{
        \begin{array}{c}
           U_m
        \end{array}
    }{U_{n_2}, \outerExt( g(\tY), \innerExt(f(S(U_{n_2}, \tY_\ell)), g(\tY)_\ell)), \tY_\ell, g(\tY)_\ell} \le \outerError \;,
    \]
    Furthermore, by applying Claim~\ref{clm:CloseToUnif} and Lemma~\ref{lma:randDistinguisher} twice, we get that
     \[
    \distCond{
            \begin{array}{c}
                 U_m, S(U_{n_2}, \tY_\ell)  \\
                 \outerExt( g(\tY), \innerExt(f(S(U_{n_2}, \tY_\ell)), g(\tY)_\ell))
            \end{array}
           }{
           \begin{array}{c}
           U_m, X\\ \outerExt( g(\tY), \innerExt(f(X), g(\tY)_\ell))
           \end{array}
           }{\begin{array}{c}
                 \tY_\ell  \\
                 g(\tY)_\ell
            \end{array}} \le \innerError \;. 
    \]
    and 
     \[
    \distCond{
    \begin{array}{c}
           \outerExt(\tY_\ell \concat \tY_r, U_{n_2}) , S(U_{n_2}, \tY_\ell) \\ \outerExt( g(\tY), \innerExt(f(S(U_{n_2}, \tY_\ell)), g(\tY)_\ell))
           \end{array}
           }{
           \begin{array}{c}
           \outerExt(\tY_\ell \concat \tY_r, \innerExt(X, \tY_\ell)), X \\ \outerExt( g(\tY), \innerExt(f(X), g(\tY)_\ell))
           \end{array}}{
            \begin{array}{c}
                 \tY_\ell  \\
                 g(\tY)_\ell
            \end{array}
           } \le \innerError \;. 
    \]
    The desired statement follows from triangle inequality.
    \end{proof}
    
    Similarly, we prove the following claim.
    
  \begin{claim}
        If $\Pr[Y \in \cY_{0,0}] \ge \tau$, then 
   \begin{equation}
    \distCond{
        \begin{array}{c}
            \outerExt( \tY_\ell \concat \tY_r,  \innerExt(X, \tY_\ell) )
        \end{array}
    }{
        \begin{array}{c}
           U_m
        \end{array}
    }{
        \begin{array}{c}
        \tY_\ell\\ g(\tY)_\ell\\ X \\ \outerExt( g(\tY)  \innerExt(f(X), g(\tY)_\ell) )
        \end{array}
    } \leq 
    \outerError +  
         3\innerError +
         \sqrt{\collisionError}
       \;,
    \end{equation}
    where $\tY = Y|_{Y \in \cY_{0,0}}$. 
    \end{claim}
    \begin{proof}
        Notice that  the probability that $\innerExt(X, \tY) = \innerExt(f(X), g(\tY)_\ell)$ is at most $\sqrt{\collisionError}$. Thus, by Claim~\ref{clm:CloseToUnif}, the probability that $U_{n_2} = \innerExt(S(U_{n_2}, \tY_\ell), g(\tY)_\ell)$ is at most $\sqrt{\collisionError} + \innerError$.   Also, since $U_{n_2}$ is independent of $\tY_r$ given $\tY_\ell, g(\tY)_\ell$, and $H_\infty(U_{n_2}) = n_2 \ge k_2$, $H_\infty(\tY_r|\tY_\ell, g(\tY_r)) \ge k^* - \log 1/\tau - 2n_4 \ge k_1$, by the definition of a strong non-malleable extractor, we have that
        
    \[
    \distCond{
        \begin{array}{c}
            \outerExt( \tY_\ell \concat \tY_r,  U_{n_2} )
        \end{array}
    }{
        \begin{array}{c}
           U_m
        \end{array}
    }{
    \begin{array}{c}
    U_{n_2}\\
    \outerExt( g(\tY), \innerExt(f(S(U_{n_2}, \tY_\ell)), g(\tY)_\ell))\\
    \tY_\ell\\
    g(\tY)_\ell
    \end{array}}
    \le \outerError + \innerError + \sqrt{\collisionError} \;. 
    \]
    Furthermore, by applying Claim~\ref{clm:CloseToUnif} and Lemma~\ref{lma:randDistinguisher} twice, we get that
     \[
    \distCond{
            \begin{array}{c}
                 U_m\\ S(U_{n_2}, \tY_\ell)\\
                 \outerExt( g(\tY), \innerExt(f(S(U_{n_2}, \tY_\ell)),
                 g(\tY)_\ell))
            \end{array}
             }{
             \begin{array}{c}
             U_m\\ X\\
             \outerExt( g(\tY), \innerExt(f(X), g(\tY)_\ell))\\
             \end{array}
              }{\tY_\ell, g(\tY)_\ell} \le \innerError \;. 
    \]
    and 
     \[
    \distCond{
    \begin{array}{c}
           \outerExt(\tY_\ell \concat \tY_r, U_{n_2}) , S(U_{n_2}, \tY_\ell) \\ \outerExt( g(\tY), \innerExt(f(S(U_{n_2}, \tY_\ell)), g(\tY)_\ell))
           \end{array}
           }{
           \begin{array}{c}
           \outerExt(\tY_\ell \concat \tY_r, \innerExt(X, \tY_\ell)), X \\ \outerExt( g(\tY), \innerExt(f(X), g(\tY)_\ell))
           \end{array}}{\begin{array}{c}
                 \tY_\ell  \\
                 g(\tY)_\ell
            \end{array}} \le \innerError \;. 
    \]
    The desired statement follows from triangle inequality.
    \end{proof}
    We then conclude the proof of right strongness of our non-malleable extractor exactly as we obtained left strongness. 
\end{proof}

\begin{remark}\label{multitamper1}
    We remark that one can apply the above compiler to multi-tampering non-malleable extractors as a $\outerExt$. Briefly speaking $t$-tampering non-malleable extractor guarantees that extraction output remains uniform even given not one but $t$ tampering outputs:
    \begin{equation*}
        \distCond{E(X, Y)}{\unif{m}}{E(f_1(X), g_1(Y)),...,E(f_t(X), g_t(Y))} \leq \eps.
    \end{equation*}
    As a result, compiled extractor will also be $t$-tamperable.
    The proof is almost identical, there are only two differences:
    \begin{enumerate}
        \item To ensure the reduction to split state tampering it is not sufficient to reveal $\tY_\ell$ and $g_1(\tY)_\ell$, but also all other tamperings: $g_2(\tY)_\ell, \ldots, g_t(\tY)_\ell$. This will have an impact of the calculations of entropy requirement.
        \item Notice that when considering the collision resistance adversary has now $t$ chances instead of $1$, but since the attempts are non-adaptive we can easily bound the collision probability by $t\cdot \collisionError$, this impacts the error calculations.
    \end{enumerate}
\end{remark}

\section{Collision resistance of Extractors}
\label{sec:coll-res}
\subsection{Generic Collision Resistance for Seeded Extractors}
\label{sec:coll-res-generic}
\begin{lemma}\label{lma:CRSeededExt}
    Let $\extractorParam{\ext}{(n, k)}{(d, d)}{m}{\eps}$ be a strong seeded extractor. Then there exists a strong seeded extractor $\extractorParam{\crTreExt}{(n, k)}{(d + z, d + z)}{m - 2\log(1/\collisionError)}{\eps + \treError + \sqrt{\collisionError}}$ with collision probability $\collisionError$ and $z = O(\log(1/\collisionError)\log^2(\log(1/\collisionError))\log(1/\treError))$.
\end{lemma}
\begin{proof}
    We will first mention \cite{RRV99}'s construction of $\treExt$. The aforementioned construction uses an error correcting code and a weak design, defined respectively as below:

    \begin{lemma}[Error Correcting Code, Lemma 35 of \cite{RRV99}]
        For every $n \in \bN$, and $\delta > 0$, there exists a code $\ecc : \bit^n \to \bit^{\hat{n}}$ where $\hat{n} = poly(n, 1/\delta)$ such that for $x, x' \in \bit^n$ with $x \neq x'$, it is the case that $\ecc(x)$ and $\ecc(x')$ disagree in at least $(\frac{1}{2} - \delta) \hat{n}$ positions.
    \end{lemma}

    \begin{definition}[Weak Design, Definition 6 of \cite{RRV99}]
        A family of sets $\blocks{S}{m} \subseteq [d]$ is a weak $(\ell, \rho)$-design if:
        \begin{enumerate}
            \item For all $i$, $\lvert S_i \rvert = \ell$;
            \item For all $i$,
            \begin{equation*}
                \sum_{j < i} 2^{\lvert S_i \cap S_j \rvert} \leq \rho \cdot (m - 1).
            \end{equation*}
        \end{enumerate}
    \end{definition}
    In particular, any family of disjoint sets $\blocks{S}{m} \subseteq [d]$ with $\lvert S_i \rvert = \ell$ is trivially a weak design as well.

    Extractor $\treExt$ operates in the following way: $X$ is firstly evaluated on an error correcting code $\ecc$ to obtain $\hat{X}$. Then viewing seed bits $Z$ as $Z_1 \concat Z_2 \concat \ldots \concat Z_d$, then the $i^{th}$ bit of $\treExt(X, Z)$ is given as the $(Z_{|S_i})^{th}$ bit of $\hat{X}$ where $Z_{|S_i}$ is understood to specify an $\ell$-bit index $Z_{j_1} \concat Z_{j_2} \concat \ldots \concat Z_{j_\ell}$ for $S_i = \{ j_1, j_2, \ldots, j_\ell \}$. In short, the output is given as:

    \begin{equation*}
        \treExt(X, Z) = \hat{X}(Z_{|S_1}) \concat \hat{X}(Z_{|S_2}) \concat \cdots \concat \hat{X}(Z_{|S_m}).
    \end{equation*}

    The modification is to truncate the output of $\ext(X, S)$ by $t = \frac{5}{2}\log(1/\collisionError)$ bits, and then treating $Z$ as $\frac{4t}{5}$ blocks of $\ell = O(\log^2(t)\log(1/\treError))$ many bits, we concatenate the output with $\frac{4t}{5}$ bits. In short, the output is given as:

    \begin{equation*}
        \crTreExt(X, S \concat Z)_i = \begin{cases}
                                        \ext(X, S)_i &,\text{if } i \leq m - t\\
                                        \hat{X}(Z_{i - (m - t)}) &, \text{if } i > m - t
                                    \end{cases}
    \end{equation*}
    where $Z_j$ denotes the $j^{th}$ block of $Z$.

    To show that $\crTreExt$ is indeed a strong extractor, note that $S$ and $Z$ are independent and furthermore by Lemma \ref{lem:avgminH} $\avgCondMinEnt{X}{\ext(X, S),S} \geq k - m + t \geq t$. Instantiating $\crTreExt$ with a family of disjoint sets, an error correcting code $\ecc$ with minimum distance $(\frac{1}{2} - \frac{\treError}{4m}) \hat{n}$ for inputs of min-entropy $(t, \frac{4t}{5})$ and seed length $O( \log(1/\collisionError)\log^2(t) \log(1/\treError))$, Lemma \ref{lma:avgCaseMinExtraction} implies that:

    \begin{equation*}
        \distCond{ \ext(X, S) \treExt(X, Z) }{ \ext(X, S), \unif{\Omega(t)} }{S, Z} \leq \treError + 2^{-\frac{t}{5}}
    \end{equation*}
    which in turn yields us:
    \begin{equation*}
        \distCond{ \ext(X, S) \concat \treExt(X, Z) }{ \unif{m - O(t)} }{S, Z} \leq \eps + \treError + 2^{-\frac{t}{5}} = \eps + \treError + \sqrt{\collisionError}
    \end{equation*}

    As for the collision probability, note that for any $x$ and fixed-point-free function $f$:
    \begin{align*}
        \Pr[\crTreExt(x, S \concat Z) = \crTreExt(f(x), S \concat Z)] 
        &\leq \Pr\big[\forall i, \ecc(x)(Z_i) = \ecc(f(x))(Z_i)\big]\\
        &\leq \left(\frac{1}{2} + \frac{\treError}{4m}\right)^{2\log(1/\collisionError)}\\
        &\leq \collisionError
    \end{align*}

    Since this bound holds for all possible values $x$, it follows that it holds for any random variable $X$ as well.
\end{proof}

An instantiation that will suit our purpose will be to use Trevisan's extractor $\treExt$ as $\ext$. Then for any $n, k$, we have $\extractorParam{\crTreExt}{(n, k)}{(d, d)}{\Omega(k)}{3\treError}$ with $d = O( \log^2(n) \log(1/\treError)) + O(\log(1/\treError) \log^2(\log(1/\treError))\log(1/\treError)) = O(\log^2(n)\log^2(1/\treError) )$ such that $\eps = \treError$ andwith collision probability $\collisionError = (\treError)^2 < 2^{-\Omega(k)}$. 

\subsection{Collision Resistance of the Raz Extractor}
\label{sec:coll-res-raz}
\begin{lemma}\label{lma:CRRaz}
    For any $n_1, n_2, k_1, k_2, m$ and any $0 < \delta < \frac{1}{2}$ such that:
    \begin{enumerate}
        \item $k_1 \geq 12\log(n_2 - k_2) + 15$,
        \item $n_2 \geq 6 \log n_2 + 2\log n_1 + 4$,
        \item $k_2 \geq (\frac{1}{2} + \delta) \cdot n_2 + 3\log n_2 + \log n_1 + 4$,
        \item $m = \Omega(\min \{ n_2, k_1 \})$,
    \end{enumerate}
    there exists a strong two-source extractor $\extractorParam{\razExt}{(n_1, k_1)}{(n_2, k_2)}{m}{\eps}$, such that $\eps = 2^{-\frac{3m}{2}}$ with collision probability $2^{-m + 1}$.
\end{lemma}

\begin{proof}
We will show that the two-source extractor by Raz satisfies the collision resistant property. We first recap \cite{Raz05}'s construction. Given independent sources $X \givenBy (n_1, k_1)$ and $Y \givenBy (n_2, k_2)$, $\razExt(X, Y)$ uses $Y$ as seed (using Lemma~\ref{lma:linearTests}) to construct $m \cdot 2^{n_2}$ many $0$-$1$ random variables $Z_{(i, X)}(Y)$ with $i \in [m]$ and $x \in \bit^{n_1}$, where random variables are $(t', \eps)$-biased for $t' \geq t \cdot m$.

The idea is to generate a sequence random variables are $\eps$-biased for tests of size $2tm$, and then the probability of collision can be bounded in a similar manner as the proof that function is a two-source extractor. Define $\gamma_i(X, Y) = (-1)^{Z_{i, X}(Y)}$ and let $f$ be any fixed-point-free function. Furthermore let $t' \geq 2 \cdot m t$ for some value of $t$ such that the set of random variables $Z_{(i, x)}(Y)$ are $(t', \eps)$-biased. The idea will be to show that we can leverage the $(t, \eps)$-biasedness to show that with high probability over the choice of $X$, for each $i \in [m]$, the probability of the extractor colliding on the $i^{th}$ bit is close to $1/2$. Then we use the Lemma \ref{lma:XORLemma} to argue that overall the probability of colliding on all bits is small. 

More formally, define $\gamma_i(X, Y) = \bE\left[ (-1)^{Z_{i, X}(Y)} \right]$, and let $f$ be any fixed-point free function. We will first bound $\lvert \gamma_i(X, Y) \rvert$. 

\begin{claim}[Claim 3.2 in \cite{Raz05}]\label{clm:invokeEpsBias}
    For any $i \in [m]$, any $r \in [t']$ and any set of distinct values $\blocks{x}{r} \in \bit^{n_1}$:
    \begin{equation*}
        \sum_{y \in \bit^{n_2}} \prod^r_{j = 1} (-1)^{Z_{i, x_j}(y)} \leq 2^{n_2} \cdot \eps
    \end{equation*}
\end{claim}

\begin{proof}
    Since $Z_{i, x}$ are $(t', \eps)$-biased:
    \begin{align*}
        \sum_{y \in \bit^{n_2}} (-1)^{Z_{i, x_j}(y)} &= \sum_{y \in \bit^{n_2}} (-1)^{\bigoplus_{j} Z_{i, x_j}(y)}
        = 2^{n_2} \sum_{y \in \bit^{n_2}} \Pr[\unif{n_2} = y] (-1)^{\bigoplus_{j} Z_{i, x_j}(y)}\\
        &= 2^{n_2} (-1)^{\bigoplus_{j} Z_{i, x_j}(\unif{n_2})} \leq 2^{n_2} \cdot \eps
    \end{align*}
\end{proof}

\begin{claim}\label{clm:pointwiseBiased}
    Letting $Z_{(i, x)}(Y)$ be $(2t, \eps)$-biased, $\Pr[ Z_{(i, x)}(Y) = Z_{(i, f(x))}(Y) ] = \Pr[ Z_{(i, x)}(Y) \oplus Z_{(i, f(x))}(Y) = 0] \leq \frac{1}{2} + \eps'$ where:
    \begin{equation*}
        \eps' = 2^{(n_2 - k_2) / t} \cdot \left( \eps^{1/t} + (2t) \cdot 2^{-\frac{k_1}{3}} \right)
    \end{equation*}  
\end{claim}
\begin{proof}
    Let $t$ be some even positive integer, then consider $\left( \gamma(X, Y) \gamma(f(X), Y) \right)^t$. By Jensen's inequality we can bound the term as:
    \begin{align*}
        &\left( \gamma(X, Y) \gamma(f(X), Y) \right)^t =  \left( \frac{1}{2^{k_1 + k_2}} \sum_{(x, y) \in \supp{X, Y}} \ (-1)^{Z_{(i, x)}(y) \oplus Z_{(i, f(x))}(y)} \right)^t \\
          &\leq  \left(\frac{1}{2^{k_2}}\right) \sum_{y \in \supp{Y}} \left( \frac{1}{2^{k_1}} \sum_{x \in \supp{X}} (-1)^{Z_{(i, x)}(y) \oplus Z_{(i, f(x))}(y)}  \right)^t \\
      &\leq  \left(\frac{1}{2^{k_2}}\right) \sum_{y \in \bit^{n_2}} \left( \frac{1}{2^{k_1}} \sum_{x \in \supp{X}} (-1)^{Z_{(i, x)}(y) \oplus Z_{(i, f(x))}(y)}  \right)^t \\
        &=  \left(\frac{1}{2^{k_2 + k_1 \cdot t}}\right) \sum_{\blocks{x}{t} \in \supp{X}} \sum_{y \in \bit^{n_2}}  \prod_{j = 1}^t  (-1)^{Z_{(i, x_j)}(y) \oplus Z_{(i, f(x_j))}(y)}  \\
    \end{align*}
  
    Then we partition the summands (based on $\blocks{x}{t}$) into two categories: (1) When the values $\blocks{x}{t}, \functionBlocks{x}{t}{f}$ has at least one unique value $x$ that does not otherwise occur in $\blocks{x}{t}$ and $\functionBlocks{x}{t}{f}$ or else (2) when the every value in $\blocks{x}{t}, \functionBlocks{x}{t}{f}$ occurs at least twice.
    
    (1) In the first case, Claim \ref{clm:invokeEpsBias} implies the respective summands can be bounded by $2^{n_1} \cdot \eps$ and there are at most $2^{k_1 \cdot t}$ many of these summands. (2) In the latter case, we will bound the sum using the following claim:

    \begin{claim}
        If $\blocks{x}{t}, \functionBlocks{x}{t}{f}$ are such that every value occurs at least twice and $f(x_i) \neq x_i$ for all $i \in [t]$, then there exists a subset of indices $S \subseteq [t]$ such that $\lvert S \rvert \leq \frac{2}{3} t $ and $\{ \blocks{x}{t} \} \subseteq \{ x : s \in S \} \cup \{ f(x) : s \in S \}$.
    \end{claim}

    \begin{proof}
        Define $A$ to contain the of values of \blocks{x}{t} that occur at least twice within \blocks{x}{t}. Define $S_A$ be the set of indices of the first occurrence of each value in $A$, and furthermore define $B$ to be $\{ \blocks{x}{t} \} \setminus \{ x_j, f(x_j) : j \in S_A \}$. Then if $\lvert A \rvert = \ell$, $\lvert B \rvert = r \leq t - 2\ell$. Let $B = \{b_1, \ldots, b_r\}$.
        
        Since each $\blocks{x}{t}, \functionBlocks{x}{t}{f}$ has that every value occurs twice, and $b_i$ for any $i \in [r]$ does not occur in $\{x, f(x) \: : \: x \in S_A\}$, it implies that $\blocks{b}{r} \in B$ must be a fixed-point-free permutation of $\functionBlocks{B}{r}{f}$. Thus, the permutation $f$ defines a disjoint union of cycles over the set $B$. Define $S_B$ to be the set that for each such cycle includes every alternate element. More precisely, for each such cycle, say $(b_{i_1}, \ldots, b_{i_q})$ with \[
        f(b_{i_1}) = b_{i_2}, f(b_{i_2}) = b_{i_3}, \ldots, f(b_{i_{q -1}}) = b_{i_q}, f(b_{i_q}) = b_{i_1} \;, \] we include $b_{i_1}, b_{i_3}, \ldots, b_{i_{1 + 2 \lfloor (q-1)/2\rfloor}}$ in the set $S_B$. Then $S = S_A \cup S_B$ satisfy the desired condition. Also, \[|S_B| \le r \max_{q \in \mathbb{N} \setminus \{1\}} \frac{\lceil q/2 \rceil}{q} \le \frac{2r}{3}\;,\] since $\frac{\lceil q/2 \rceil}{q}$ is $1/2$ when $q$ is even, and $(q+1)/2q$ when $n$ is odd, and hence is maximized for $q = 3$.  Thus, 
        \[
        |S| \le \ell + \frac{2r}{3} \le \ell + \frac{2(t-2\ell)}{3} = \frac{2t}{3} - \frac{\ell}{3} \le \frac{2t}{3} \;,
        \]
        as needed.
    \end{proof}

    To obtain the bound on the number of summands in the case (2), note that there are $\binom{2^{k_1}}{\frac{2}{3}t }$ possible sets $S$, and for each set, there are $\left( \frac{4t}{3}\right)^t$ possible sequences that satisfy Case 2. In each such case, we bound the summand by $2^{n_2}$. Combining the two cases, we get that:

    \begin{align*}
        &\left( \gamma(X, Y) \gamma(f(X), Y) \right)^t \leq \left(\frac{1}{2^{k_2 + k_1 \cdot t}}\right) \sum_{\blocks{x}{t} \in \supp{X}} \sum_{y \in \bit^{n_2}}  \prod_{j = 1}^t  (-1)^{Z_{(i, x_j)}(y) \oplus Z_{(i, f(x_j))}(y)}  \\
        &\leq \left(\frac{1}{2^{k_2 + k_1 \cdot t}}\right) \left( 2^{k_1 \cdot t} 2^{n_2} \cdot \eps + 2^{n_2}\binom{2^{k_1}}{ \frac{2}{3}t }  \left( \frac{4t}{3}\right)^t \right)\leq \left(\frac{1}{2^{k_2 + k_1 \cdot t}}\right) \left( 2^{k_1 \cdot t} 2^{n_2} \cdot \eps + 2^{n_2} (2t)^t \cdot 2^{-\frac{k_1}{3}t} \right)\\
        &\lvert \gamma(X, Y) \gamma(f(X), Y) \rvert \leq 2^{(n_2 - k_2) / t} \cdot \left( \eps^{1/t} + (2t) \cdot 2^{-\frac{k_1}{3}} \;. \right)
    \end{align*}
\end{proof}

Now that we have shown that for any coordinate $i \in [m]$, the probability the extractor collides on the $i^{th}$ bit is at most $\frac{1}{2} + \eps'$, we wish to invoke the Lemma \ref{lma:XORLemma} to argue that the probability the extractor collides on all the coordinates is small.

Define $\tau \subseteq [m]$, and consider the set of random variables\\ $\left \{ \bigoplus_{i \in \tau} Z_{i, x}(Y) \oplus \bigoplus_{i \in \tau} Z_{i, f(x)}(Y) : x \in \bit^{n_1} \right\}$. Since $\lvert \tau \rvert \leq m$, the set of random variables is $\eps$-biased for linear tests of size up to $\frac{2t'}{m}$, and hence $\bigoplus_{i \in \tau} Z_{i, x}(Y) \oplus \bigoplus_{i \in \tau} Z_{i, f(x)}(Y)$ is $\eps'$-biased by Claim \ref{clm:pointwiseBiased}. Then by the Lemma \ref{lma:XORLemma}, since this holds for any $\tau \subseteq [m]$, the sequence $( Z_{1, X}(Y) \oplus Z_{1, f(X)}(Y),\ldots, Z_{m, X}(Y) \oplus Z_{m, f(X)}(Y))$ is $\eps' \cdot 2^{\frac{m}{2}}$-close to $\unif{m}$. It follows that, the probability of collision is at most:
\begin{equation*}
    2^{-m} + \eps' \cdot 2^{\frac{m}{2}} = 2^{-m} + 2^{\frac{m}{2}} \cdot 2^{(n_2 - k_2) / t} \cdot \left( \eps^{1/t} + (2t) \cdot 2^{-\frac{k_1}{3}} \right) \;.
\end{equation*}

We now bound the probability of collision based on our choice of parameters. Recall that Lemma \ref{lma:linearTests} asserts that we can construct $m \cdot 2^{n_2}$ many variables $Z_{(i, x)}$ that are $(t', \eps)$-biased using $2 \lceil \log(1/\eps) + \log \log (m 2^{n_2})   + \log(t') \rceil = 2 \lceil \log(1/\eps) + \log \log (m 2^{n_2}) + \log(2mt) \rceil$ random bits. Set $\eps = 2^{-r}$ where $r = \frac{1}{2} n_2 + 3\log n_2 + \log n_1$, $n_2 \geq 16$ and $k_1 \geq 64$. We then bound the probability separately depending on $k_1$'s value relative to $4(n_2 - k_2)$.

\textbf{If $k_1 \leq 4(n_2 - k_2)$:} Choose $t$ to be the smallest even integer such that $t \geq \frac{8 (n_2 - k_2)}{k_1}$. Then $t \leq n_2 - k_2$, or else that would imply that $k_1 \leq 8$. Then it follows that:
\begin{equation*}
    \frac{8(n_2 - k_2)}{k_1} \leq t \leq \frac{16(n_2 - k_2)}{k_1} \leq \frac{8 n_2}{k_1}
\end{equation*}

Using the inequality above:
\begin{align*}
    2^{(n_2 - k_2) / t} &\cdot \left( \eps^{1/t} + (2t) \cdot 2^{-\frac{k_1}{3}} \right)
    \leq 2^{(n_2 - k_2 - r) / t} + \frac{32(n_2 - k_2)}{k_1} 2^{-\frac{k_1}{3}}\\
    &\leq 2^{-\delta n_2 / t} + \frac{32(n_2 - k_2)}{k_1} 2^{-\frac{k_1}{3}}
    \leq 2^{-\delta n_2 / t} + 2^{-\frac{k_1}{3} + \frac{k_1}{12}} \leq 2^{-\delta \frac{k_1}{8}} + 2^{-\frac{k_1}{4}} \leq 2^{-\delta \frac{k_1}{8} + 1}\\
\end{align*}

\textbf{Otherwise, if $k_1 > 4(n_2 - k_2)$:} Set $t = 2$. Then:

\begin{align*}
    2^{(n_2 - k_2) / 2} \cdot \left( \eps^{1/2} + 4 \cdot 2^{-\frac{k_1}{3}} \right) &= 2^{(n_2 - k_2 - r) / 2} + 2^{(n_2 - k_2 )/2} \cdot 4 \cdot 2^{-\frac{k_1}{3}}\\
    &\leq 2^{-\delta n_2 / 2} + 2^{(n_2 - k_2 )/2} \cdot 4 \cdot 2^{-\frac{k_1}{3}}
    \leq 2^{-\delta n_2 / 2} + 2^{-\frac{k_1}{8}}
\end{align*}

Choosing $m \leq \delta \min\{ \frac{n_2}{4}, \frac{k_1}{16} \} - 1$, we get that the collision probability is at most $2^{-m} + 2^{\frac{m}{2} - 2m - 1} \leq 2^{-m + 1}$.
\end{proof}

\section{A Fully Non-malleable Seeded Extractor}
\label{sec:fullynmseeded}
In this section, we will use $\crTreExt$ as $\innerExt$ and $\liExt$ as $\outerExt$ for Theorem \ref{thm:main_reduction} with the following instantiations:

\begin{enumerate}
    \item $\crTreExt$ is an extractor given by $[(n_x, k_x), (s, s) \mapsto d \sim \treError]$ for $s = O(\log^2(n_x)\log^2(1/\treError)$, and $d = \Omega(k_x)$, with collision probability $\left(\frac{\treError}{3}\right)^2$.
    \item $\liExt$ is an extractor given by $[(d, (1-\gamma)d), (d, (1-\gamma)d) \mapsto m \sim \liError]$ for some constant $\gamma$, $m = \Omega(d)$, and $\liError = 2^{-d \left(\frac{\log \log d}{\log d} \right)}$.
\end{enumerate}
with $\collisionError = 2^{-(k_x)^c}$ for some $c < \frac{1}{2}$. It follows that $s = o(d)$.

\begin{theorem}\label{thm:FNMSExt}
    For any $n_x, k_x$, there exists a fully non-malleable seeded extractor $\extractorParam{\fnmext}{(n_x, k_x)}{(s + d, s + d)}{m}{\fnmError}$ with $m = \Omega(d)$, $d < k_x$, $s = O(\log^2(n_x)\log^2(\treError))$, $\fnmError < 10\treError$ with $\treError = 2^{-(\frac{k_x}{2})^c}$ for some $c < \frac{1}{2}$.
\end{theorem}
\begin{proof}
    It suffices to show that for our choice of parameters, the entropy requirements of $\crTreExt$ (from Lemma \ref{lma:CRSeededExt}) and $\liExt$ (from Lemma \ref{lma:li}) are met for Theorem \ref{thm:main_reduction}.

    Setting input parameters $n_3 = n_x$, $k^*_1 = k_x$, $n_4 = k_4 = s$, $k^*_2 = s + d$, and extractor parameters $n_1 = n_2 = d$, $k_1 = k_2 = (1-\gamma)d$, $k_3 = k_x$, note that indeed $k_1^* \geq k_3$. Furthermore, 
    \begin{align*}
        k^*_2 &= s + d = s + k_4 + n_1 - n_4\\
        k^*_2 &= d + s \geq \left(\frac{\gamma}{2}\right)d + (1-\gamma)d + 2s \;.
    \end{align*}

    And thus by our choice of $s$, $\fnmError \leq 3\cdot 2^{-(\frac{k_x}{2})^{2c}} + 7\treError < 10 \treError$ with $\treError = 2^{-(\frac{k_x}{2})^c}$ for some $c < \frac{1}{2}$.
\end{proof}
    
It will also be useful in the subsequent subsection that we relax the entropy requirement of this extractor.
\begin{theorem}\label{thm:relaxedFNMSExt}
    For any $n_x, k_x$, there exists a fully non-malleable seeded extractor $\extractorParam{\fnmext}{(n_x, k_x)}{(s + d, s + d - 1)}{m}{\fnmError}$ with $m = \Omega(k_x)$, $d < k_x$, $s = O(\log^2(n_x)\log^2(\treError))$, $\fnmError < 12\treError$ with $\treError = 2^{-(\frac{k_x}{2})^c}$ for some $c < \frac{1}{2}$.
\end{theorem}
    
\begin{proof}
    By Lemma \ref{lma:extReduction} and Lemma \ref{lma:colReduction}, $\crTreExt$ can also be viewed as $\extractorParam{\crTreExt}{(n_x, k_x)}{(s, s - 1)}{\Omega(k_x)}{2 \treError}$ with collision probability $2 \collisionError = 2\left(\frac{\treError}{3} \right)^2 \leq \treError^2$. 
    
    For a similar choice of parameters: $n_3 = n_x$, $k^*_1 = k_x$, $n_4 = s$, $k_4 = s - 1$ $k^*_2 = s + d$, and extractor parameters $n_1 = n_2 = d$, $k_1 = k_2 = (1-\gamma)d$, $k_3 = k_x$, note that indeed $k_1^* \geq k_3$. Furthermore, 
    \begin{align*}
        k^*_2 &= s + d - 1 = s + s - 1 + d - s = s + k_4 + n_1 - n_4\\
        k^*_2 &= s + d - 1 \geq \left(\frac{\gamma}{2}\right)d + (1-\gamma)d + 2s - 1
    \end{align*}

    And thus by our choice of $s$, $\fnmError \leq 3\cdot 2^{-(\frac{k_x}{2})^{2c}} + 9\treError < 12 \treError$ with $\treError = 2^{-(\frac{k_x}{2})^c}$ for some $c < \frac{1}{2}$.
\end{proof}
\begin{remark} \label{multitamper2}
    As we have already mentioned in the Remark \ref{multitamper1}, we can use $t-$tamperable extractor like \cite{Li17}. As a result our $\fnmext$ will be $t$-tamperable non-malleable extractor with negligible error. One only has to make sure that $|y_\ell|<\frac{\gamma \cdot n }{t+1}$, which follows from the first point in the Remark \ref{multitamper1}, where $(1-\gamma)\cdot n$ is the entropy requirement from \cite{Li17} extractor. The error one obtains is therefore at least $2^{-\Omega(n/ \log n)}+2^{-\Omega(\frac{\gamma \cdot n}{t+1} - \log^2(n))}+t\cdot 2^{-\Omega (\frac{\gamma \cdot n}{t+1})}\geq 2^{-\Omega(n^c)}$, for $c<1$ depending on $t$ only. Please notice that the entropy requirements for this extractor do not change.
\end{remark}

\section{A Two-Source Non-malleable Extractor}
\label{sec:nmraz}
In this section, we will use $\razExt$ as $\innerExt$ and $\fnmext$ as $\outerExt$ from Theorem \ref{thm:relaxedFNMSExt} with the following instantiations:

\begin{enumerate}
    \item $\extractorParam{\razExt}{(n_x, k_x)}{(\YLeftLength, k_\ell)}{d}{\razError{d}}$ with $d = \Omega(\min\{k_x, k_\ell\})$ and collision probability $2^{-d + 1}$.
    \item $\extractorParam{\fnmext}{(n_y, \tau \cdot d)}{(d, d - 1)}{m}{\fnmError}$ is a two-source non-malleable extractor for some $0 < \tau < 1$, $m = \Omega(d)$, and $\fnmError < 12 \cdot \treError$ with $\treError < 2^{-\Omega((m)^c)}$ for some $c < \frac{1}{2}$.
\end{enumerate}

\begin{theorem}\label{thm:betterTwoNMExt}
    There exists a two source non-malleable seeded extractor $\twonmext : [(n_x, k_x), (n_y, k_y) \mapsto m \sim \twonmError]$, and $m = \Omega( \min\{ n_y, k_x \} )$, such that:

    \begin{enumerate}
        \item $k_x \geq 12 \log(n_y - k_y) + 15$,
        \item $n_y \geq 30 \log(n_y) + 10\log(n_x) + 20$,
        \item $k_y \geq (\frac{4}{5} + \gamma)n_y + 3\log(n_y) + \log(n_x) + 4$,
        \item $\twonmError \leq 3 \cdot 2^{- \frac{9\gamma}{10} n_y} + 40 \cdot \treError$ where $\treError = 2^{-\Omega(d^c)}$ with $c < \frac{1}{2}$.
    \end{enumerate}
\end{theorem}

\begin{proof}
For any given $Y \givenBy (n_y, k_y)$, we treat it as $Y = \YLeft \concat \YRight$ where $\lvert \YLeft \rvert = n_\ell$ and $\lvert \YRight \rvert = n_r$. 

The extractor $\extractorParam{\razExt}{(n_x, k_x)}{(\YLeftLength, k_\ell)}{d}{\razError{d}}$ from Lemma $\ref{lma:CRRaz}$ requires the following conditions:
\begin{enumerate}
    \item $k_x \geq 12\log(\YLeftLength - k_\ell) + 15$
    \item $\YLeftLength \geq 6 \log \YLeftLength + 2\log n_x + 4$,
    \item $k_\ell \geq (\frac{1}{2} + \gamma) \cdot \YLeftLength + 3\log \YLeftLength + \log n_x + 4$,
    \item $d \leq \gamma \min\{ \frac{n_\ell}{4}, \frac{k_x}{16} \} - 1$
\end{enumerate}
for some $0 < \gamma < \frac{1}{2}$.

Setting $\YLeftLength = (\frac{2}{5} - \gamma) n_y$ (and consequently $n_r = (\frac{3}{5} + \gamma)n_y$), we first show that indeed the input requirements for $\razExt$ are met. Note that 
\begin{align*}
    (n_y - k_y) - (n_\ell - k_\ell) = n_y - k_y - (n_\ell - (k_y - n_r)) = 0
\end{align*}
which implies that:
\begin{align*}
    k_x \geq 12 \log(n_y - k_y) + 15 = 12 \log(n_\ell - k_\ell) + 15
\end{align*}
Next:
\begin{align*}
    n_\ell \geq \frac{1}{5} n_y \geq 6 \log(n_y) + 2\log(n_x) + 4 \geq 6 \log(n_\ell) + 2\log(n_x) + 4
\end{align*}
And lastly:
\begin{align*}
    k_\ell \geq k_y - n_r &= \left(\frac{4}{5} + \gamma\right)n_y + 3\log(n_\ell) + \log(n_x) + 4 - \left(\frac{3}{5} + \gamma\right) n_y\\
    &= \left(\frac{1}{5}\right)n_y + 3\log(n_\ell) + \log(n_x) + 4
    = \left(\frac{1}{5}\right) \left(\frac{1}{0.4 - \gamma}\right) n_\ell + 3\log(n_\ell) + \log(n_x) + 4\\
    &\geq \left(\frac{1}{2} + \frac{5\gamma}{4}\right) n_\ell + 3\log(n_\ell) + \log(n_x) + 4\\
\end{align*}

Setting input parameters $n_3 = n_x$, $k_1^* = k_x$, $n_1 = n_y$, $k_2^* = (\frac{4}{5} + \gamma) n_y$, and extractor parameters $n_4 = n_\ell$, $k_4 = k_\ell$, $n_1 = n_y$, $k_1 = \tau \cdot d$, $n_2 = d$, $k_2 = d - 1$ for some $0 < \tau < 1$, we get that $k_1^* \geq k_3$. Furthermore:
\begin{align*}
    k_2^* - k_4 - n_1 + n_4 &= k_y  - k_\ell - n_y + n_\ell = k_y - k_\ell - \left(\frac{3}{4} + \gamma\right) n_y \\
    &\geq \left( \frac{1}{5} + \gamma \right) n_y - \left( \frac{1}{2} + \gamma \right)\left( \frac{2}{5} - \gamma \right)n_y\\
    &= \left( \frac{11}{10} \gamma + \gamma^2 \right) n_y
\end{align*}
and:
\begin{align*}
    k_2^* - k_1 - 2n_4 &= k_2^* - \tau \cdot d - 2n_\ell \geq \gamma n_y - \tau \gamma \frac{n_\ell}{4} \geq \frac{9\gamma}{10} n_y 
\end{align*}

Thus, by Theorem \ref{thm:main_reduction} it follows that $\extractorParam{\twonmext}{(n_3, k_1^*)}{(n_1, k_2^*)}{m}{ \twonmError }$ is a strong non-malleable extractor with error:
\begin{align*}
    \twonmError &\leq 3 \cdot 2^{- \frac{9\gamma}{10} n_y} + 36 \cdot \treError + 2 \cdot 2^{-\frac{3}{2} d} + 2 \sqrt{2^{-d + 1}}\\
                &\leq 3 \cdot 2^{- \frac{9\gamma}{10} n_y} + 40 \cdot \treError
\end{align*}
where $\treError = 2^{-\Omega(d^c)}$ with $c < \frac{1}{2}$.
\end{proof}
\begin{remark}\label{multitamper3}
    As we noted in Remark \ref{multitamper1} we can use multi-tampering extractor from Remark \ref{multitamper2}, and obtain a $t$-tamperable non-malleable extractor. The error of such extractor remains negligible. Entropy requirements change due to first point from Remark \ref{multitamper1}: One source can have poly-logarithmic entropy, while the other requires entropy rate $(1-\frac{1}{2t+3})$. 
\end{remark}

\section{A Two-Source Non-malleable Extractor With Rate $\frac{1}{2}$}
\label{sec:ratehalfnm}
In \cite{AKOOS21} the authors give a compiler that turns any left-strong non-malleable extractor into a non-malleable extractor with optimal output rate of $\frac 12$. The construction looks as follows:
\begin{equation*}
    \twonmext^*(X,Y)=\sext(X, \twonmext(X,Y)),
\end{equation*}
where $\sext$ is a seeded extractor from \cite{GUV09} with output size equal $\frac 12 \minEnt{X}$, and $\twonmext$ is a left-strong non-malleable extractor.

We will briefly discuss the idea behind that construction. Let $X'$ be a tampering of $X$, and $Y'$ be a tampering of $Y$. We need to argue that if $X\neq X' \; \lor \; Y\neq Y'$ then $\twonmext^*(X,Y)$ remains uniform even given $\twonmext^*(X',Y')$. If $X\neq X'\; \lor \; Y\neq Y'$ then left-strong non-malleable extractor $\twonmext(X,Y)$ is uniform even given $\twonmext(X',Y'), X$. The final idea crucially relies on the fact that $\sext$ extracts only half of the entropy of $X$: we can reveal $\twonmext(X',Y')$ and then $\sext(X',\twonmext(X',Y'))$ becomes a leakage from $X$ (i.e. it is just a deterministic function of $X$ with a small output). We get that\\
$\avgCondMinEnt{X}{\twonmext(X',Y'),\sext(X',\twonmext(X',Y'))}\approx \frac 12 \minEnt{X}$ (size of $\twonmext(X',Y')$ is tiny so it's asymptotically irrelevant). Moreover by the left-strong property of $\twonmext$ we get that $X$ and $\twonmext(X,Y)$ remain independent given $\twonmext(X',Y'),\sext(X',\twonmext(X',Y'))$, this means that $\sext(X,\twonmext(X,Y))$ is uniform given $\twonmext(X',Y'),\sext(X',\twonmext(X',Y'))$ which gives the result.


If we make use of $\twonmext$ from the previous section we can obtain a two-source unbalanced non-malleable extractor with rate $\frac{1}{2}$.

\begin{lemma}[Theorem 5 of \cite{AKOOS21}]\label{lma:RateHalfCompiler}
    If $\extractorParam{\twonmext}{(n_1, k_1)}{(n_2, k_2)}{d}{\eps_1}$ is a strong two-source unbalanced non-malleable extractor, with $n_2 = o(n_1)$ and $\extractorParam{\ext}{(n_1, k_1)}{(d, d)}{\ell}{\eps_2}$ is a strong seeded extractor, then there exists a two source non-malleable extractor $\extractorParam{\twonmext^*}{(n_1, k_1)}{(n_2, k_2)}{\ell}{ \eps_1 + \eps_2}$. Furthermore, if $k_1, \ell < \frac{n_1}{2}$, then $\twonmext^*$ has a rate of $\frac{1}{2}$.
\end{lemma}

\begin{theorem}
    There exists an extractor $\extractorParam{\twonmext^*}{(n_1, k_1)}{(n_2, k_2)}{\ell}{\eps_1 + \eps_2}$ such that:
    \begin{enumerate}
        \item $k_1 \geq \max\{ 12 \log(n_2 - k_2) + 15, \log^3(n_1) \log(1/\eps_2) \}$
        \item $n_2 \geq \max\{ 30 \log(n_2) + 10\log(n_1) + 20, \log^3(n_1) \log(1/\eps_2) \}$
        \item $k_2 \geq (\frac{4}{5} + \gamma)n_2 + 3\log(n_2) + \log(n_1) + 4$
        \item $\eps_1 \leq 3 \cdot 2^{- \frac{9\gamma}{10} n_2} + 40 \cdot \treError$ where $\treError = 2^{-\Omega(d^c)}$ with $c < \frac{1}{2}$
        \item $\ell < \frac{k_1}{2}$
    \end{enumerate}
    Furthermore, if $n_2 = o(n_1)$, $k_1, \ell < \frac{n_1}{2}$, then $\twonmext^*$ has a rate of $\frac{1}{2}$.
\end{theorem}

\begin{proof}
    By Theorem \ref{thm:betterTwoNMExt} there exists an extractor $\extractorParam{\twonmext}{(n_1, k_1)}{(n_2, k_2)}{m}{\eps_1}$ such that:
    \begin{enumerate}
        \item $k_1 \geq 12 \log(n_2 - k_2) + 15$
        \item $n_2 \geq 30 \log(n_2) + 10\log(n_1) + 20$
        \item $k_2 \geq (\frac{4}{5} + \gamma)n_2 + 3\log(n_2) + \log(n_1) + 4$
        \item $\eps_1 \leq 3 \cdot 2^{- \frac{9\gamma}{10} n_2} + 40 \cdot \treError$ where $\treError = 2^{-\Omega(d^c)}$ with $c < \frac{1}{2}$
        \item $m = \Omega( \min\{ n_2, k_1 \} )$
    \end{enumerate}

    Using Lemma \ref{lma:tre}, $\extractorParam{\treExt}{(n_1, k_1)}{(m, m)}{\Omega(k_1)}{\eps_2}$ is a strong seeded extractor with $m = O(\log^2(n_1)\log(1/\eps_2))$. Thus by Lemma \ref{lma:RateHalfCompiler} there exists a two source non-malleable extractor $\extractorParam{\twonmext^*}{(n_1, k_1)}{(n_2, k_2)}{\Omega(k_1)}{\eps_1 + \eps_2}$.

    Furthermore, with $n_2 = o(n_1)$ and $k_1, \ell < \frac{n_1}{2}$, we get that $\twonmext^*$ has a rate of at most $\frac{n_1}{2(n_1 + n_2)} < \frac{1}{2}$.
\end{proof}

\section{Privacy Amplification against Memory Tampering Active Adversaries.}\label{App:PA}

Imagine Alice and Bob sharing some random but not uniform string $W$, they would like to "upgrade" their random string $W$ to uniformly random string. However Eve is fully controlling a channel between Alice and Bob and can arbitrarily tamper with the messages sent. The Privacy Amplification (PA) protocol guarantees that either Alice and Bob will end up with the same uniform string (unknown to Eve), or at least one of them will abort\footnote{If one of the parties, say Alice, aborts but Bob generates random string $R_B$ then we require $R_B$ to be uniform and unknown to Eve.}.

In \cite{AORSS20} the authors consider a stronger version of PA which they call a \emph{privacy amplification resilient against memory-tampering active adversaries}. In their model, Alice and Bob have access to a shared string $W$ and their local sources of (not necessarily uniform) randomness $A$ and $B$ respectively. At the beginning of the protocol Eve can select one party, say Alice, and corrupt her memory $F(W,A)=(\tilde W, \tilde A)$ (or $F(W,B)=(\tilde W, \tilde B)$ if Eve decides to corrupt Bob). If Eve did not corrupt the memory of any of the parties then the standard PA guarantees follow. On the other hand if Eve decides to corrupt one of the parties then either Alice and Bob agree on a uniformly random string (unknown to Eve) or the non-corrupted party will detect the tampering.

The following two definitions are taken verbatim from $\cite{AORSS20}$.
\begin{definition}[Protocol against memory-tampering active adversaries]
    An \emph{$(r,\ell_1,k_1,\ell_2,\linebreak[1] k_2,m)$-protocol against memory-tampering active adversaries} is a protocol between Alice and Bob, with a man-in-the-middle Eve, that proceeds in $r$ rounds.
    Initially, we assume that Alice and Bob have access to random variables $(W,A)$ and $(W,B)$, respectively, where $W$ is an $(\ell_1,k_1)$-source (the \emph{secret}), and $A$, $B$ are $(\ell_2,k_2)$-sources (the \emph{randomness tapes}) independent of each other and of $W$. The protocol proceeds as follows:
    \begin{description}
        \item In the first stage, Eve submits an arbitrary function $F:\{0,1\}^{\ell_1}\times\{0,1\}^{\ell_2}\to \{0,1\}^{\ell_1}\times\{0,1\}^{\ell_2}$ and chooses one of Alice and Bob to be corrupted, so that either $(W,A)$ is replaced by $F(W,A)$ (if Alice is chosen), or $(W,B)$ is replaced by $F(W,B)$ (if Bob is chosen).
    
   \item  In the second stage, Alice and Bob exchange messages $(C_1,C_2,\dots,C_r)$ over a non-authenticated channel, with Alice sending the odd-numbered messages and Bob the even-numbered messages, and Eve is allowed to replace each message $C_i$ by $C'_i$ based on $(C_1,C'_1,\dots,C_{i-1},C'_{i-1},C_i)$ and independent random coins, so that the recipient of the $i$-th message observes $C'_i$.
    Messages $C_i$ sent by Alice are deterministic functions of $(W,A)$ and $(C'_2,C'_4,\dots,C'_{i-1})$, and messages $C_i$ sent by Bob are deterministic functions of $(W,B)$ and $(C'_1,C'_3,\dots,C'_{i-1})$.
    
 \item    In the third stage, Alice outputs $S_A\in\{0,1\}^m\cup\{\bot\}$ as a deterministic function of $(W,A)$ and $(C'_2,C'_4,\dots)$, and Bob outputs $S_B\in\{0,1\}^m\cup\{\bot\}$ as a deterministic function of $(W,B)$ and $(C'_2,C'_4,\dots)$.

    \end{description}
    \end{definition}

\begin{definition}[Privacy amplification protocol against memory-tampering active adversaries]
    An \emph{$(r,\ell_1,k_1,\ell_2,k_2,m,\eps,\delta)$-privacy amplification protocol against memory-tampering active adversaries} is an $(r,\ell_1,k_1,\ell_2,k_2,m)$-protocol against memory-tampering active adversaries with the following additional properties:
    \begin{itemize}
        \item {\bf If Eve is passive:} In this case, $F$ is the identity function and Eve only wiretaps.
        Then, $S_A=S_B\neq \bot$ with $S_A$ satisfying
        \begin{equation}\label{eq:sa}
            S_A,C\approx_\eps U_m,C,
        \end{equation}
        where $C=(C_1,C'_1,C_2,C'_2,\dots,C_r,C'_r)$ denotes Eve's view.
        
        \item {\bf If Eve is active:} Then, with probability at least $1-\delta$ either $S_A=\bot$ or $S_B=\bot$ (i.e., one of Alice and Bob detects tampering), or $S_A=S_B\neq\bot$ with $S_A$ satisfying~\eqref{eq:sa}.
    \end{itemize}
\end{definition}

One building block of our extension is MAC:
\begin{definition}\label{def: mac}
 A family of functions $ \mac:\{0,1\}^{\gamma}\times \{0,1\}^{\tau}\rightarrow\{0,1\}^{\delta},\mathtt{Verify}: \{0,1\}^{\gamma} \times \{0,1\}^{\delta}\times \{0,1\}^{\tau} \rightarrow \{0,1\} $ is said to be a $ \mu -$secure one time message authentication code if 
	\begin{enumerate}
		\item For $ k_{a}\in_R\{0,1\}^{\tau},\ \forall \; m\in\{0,1\}^{\gamma} $, $ \Pr[\mathtt{Verify}(m,\mac_{k_{a}}(m),k_{a})=1]=1 $,\\[1mm]
		where for any $(m,t)$, $\mathtt{Verify}(m,t,k_{a}):=$
		$\begin{cases}
		1\text{ if }\ \mac(m,k_{a})=t\\
		0\text{ otherwise}
		\end{cases}$
		\item For any $ m\neq m',t,t' $, $ \Pr\limits_{k_{a}}[\mac(m,k_{a})=t|\mac(m',k_{a})=t']\leq \mu $, where $ k_{a}\in_R\{0,1\}^{\tau}$.
	\end{enumerate}
 \end{definition}

\begin{lemma}\label{lemma-mac}\cite{JKS93,DKKRS12}
	For any $ \gamma,\eps >0 $ there is an efficient $ \eps- $secure one time $\mac$ with $ \delta\leq (\log(\gamma)+\log(\dfrac{1}{\eps})) $, $ \tau\leq 2\delta $, where $ \tau,\gamma,\delta $ are key, message, tag length respectively.
\end{lemma}

In the \cite{AORSS20} protocol Alice and Bob exchange the random strings $A$ and $B$ and then locally compute $R=\twonmext(A\concat B,W)$. They then split $R$ into $3$ parts, Alice sends the first part to Bob to prove she has gotten the right output, Bob then sends the second part to Alice to do the same. If this phase was successful then last part of $R$ is the shared uniform string. Figure \ref{fig:ITPA} illustrates the protocol.

\begin{figure*}[h]
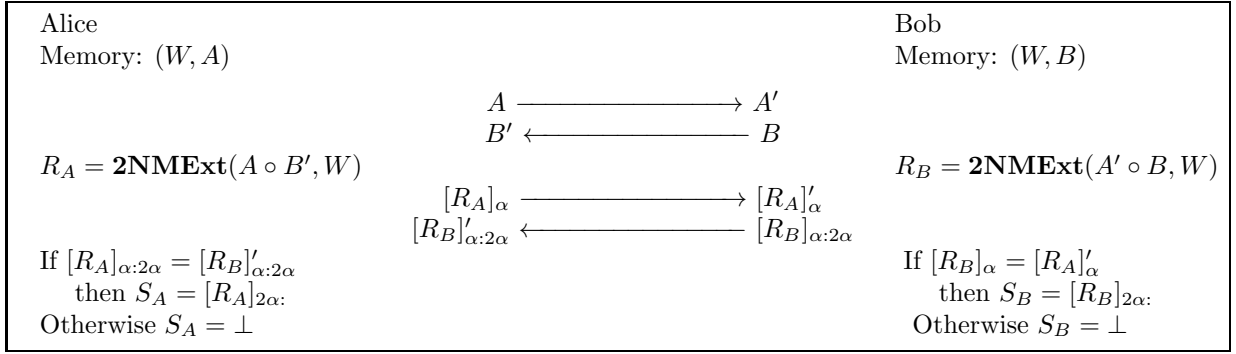

\centering
\fbox{
\begin{minipage}{15.8cm}
\begin{tabular}{lcl}
Alice & ~ &  Bob \\
Memory: $(W, A)$ & ~ &  
Memory: $(W, B)$ \vspace{0.2cm}\\

 ~ & $A$ $\xrightarrow{\quad \quad \quad \quad  \quad \quad \quad \quad}$ $A'$ ~\\
 ~ & $B'$ $\xleftarrow{\quad \quad \quad \quad  \quad \quad \quad \quad}$ $B$ ~\\
$R_A=\twonmext(A\concat B',W)$  ~ &  ~ &  $R_B=\twonmext(A'\concat B,W)$\\
 ~ & $[R_A]_\alpha$ $\xrightarrow{\quad \quad \quad \quad  \quad \quad \quad \quad}$ $[R_A]'_\alpha$ ~\\
 ~ & $[R_B]'_{\alpha : 2\alpha}$ $\xleftarrow{\quad \quad \quad  \quad \quad  \quad \quad \quad}$ $[R_B]_{ \alpha : 2\alpha}$ ~\\
If $[R_A]_{\alpha : 2\alpha}=[R_B]'_{\alpha : 2\alpha}$ & ~ & ~If $[R_B]_{\alpha}=[R_A]'_{\alpha}$ \\
$\quad \text{  then  } S_A= [R_A]_{ 2\alpha:}$ & ~ & ~$\quad \text{  then  } S_B= [R_B]_{2\alpha:}$ \\
Otherwise $S_A= \bot$ & ~& ~ Otherwise $S_B= \bot$ \\
\end{tabular}

\end{minipage}

}
\caption{Verbatim from \cite{AORSS20}. Privacy amplification protocol against memory-tampering active adversaries. In the above, for an $n$-bit string $x$ we define $[x]_i=(x_1,x_2,\dots,x_i)$, $[x]_{i:j}=(x_{i+1},\dots,x_j)$, and $[x]_{j:}=(x_{j+1},\dots,x_n)$.}
\label{fig:ITPA}
\end{figure*}

Since one of the sources of randomness might be faulty, even if the original $A, B$ were uniform, one requires a left-strong non-malleable extractor $\twonmext$ to remain secure for the first source with entropy below $0.5$, the construction of such an extractor prior to this work was unknown\footnote{Authors of \cite{AORSS20} proceed to construct a computational non-malleable extractors with parameters that would allow for this protocol to go through.}.

The above protocol obtains very short output compared to entropy of $W$, whereas ideally we would like to obtain something close to entropy of $W$.
If Alice and Bob have access to uniform randomness, one can extend this protocol to output almost as many bits as $W$'s entropy (see Figure \ref{fig:ITPA2}). After the execution of the \cite{AORSS20} protocol we have the additional guarantee (see proof of Theorem 6, point (b)) that if $S_A\neq \bot $ and $S_B\neq \bot $ then we know that $S_A = S_B$ and are close to uniform and moreover Eve did not tamper with $W$ of either of the parties (this is only achieved with standard notion of non-malleability, not the one from~\cite{GSZ21}). If Alice and Bob have access to some extra uniform bits (if $A$ and $B$ were uniform to start with then we could cut them in half $A=A_1\concat A_2$ and $B=B_1\concat B_2$, use the first half to run the original protocol by \cite{AORSS20} and save the other half for later) then we can continue the protocol (in the spirit of~\cite{DW09}): 
Alice will send $A_2,\sigma_A$ to Bob, where $\sigma_A$ is a Message Autentication Code of $A_2$ with first half of $S_A$ as a key. Bob will do the same: send $B_2, \sigma_B$ to Alice using other half of $S_B$ as a MAC key. There is a one final problem, we know that one of $A_2$ or $B_2$ is uniform but we don't know which (Eve could have left $W$ unchanged but could have tampered with random coins $A$ and $B$), moreover one of them might depend on $W$. Notice that $A_2$ and $B_2$ will remain independent, and one of them is independent of $W$ and uniform. Therefore $A_2+B_2$ is uniform and independent of $W$. Now all we have to do is plug in $W$ and $A_2+B_2$ into seeded extractor $\sext(W,A_2+B_2)$ and we can extract almost whole entropy out of $W$ (and the output remains hidden from the view of Eve).


\begin{figure*}[h]
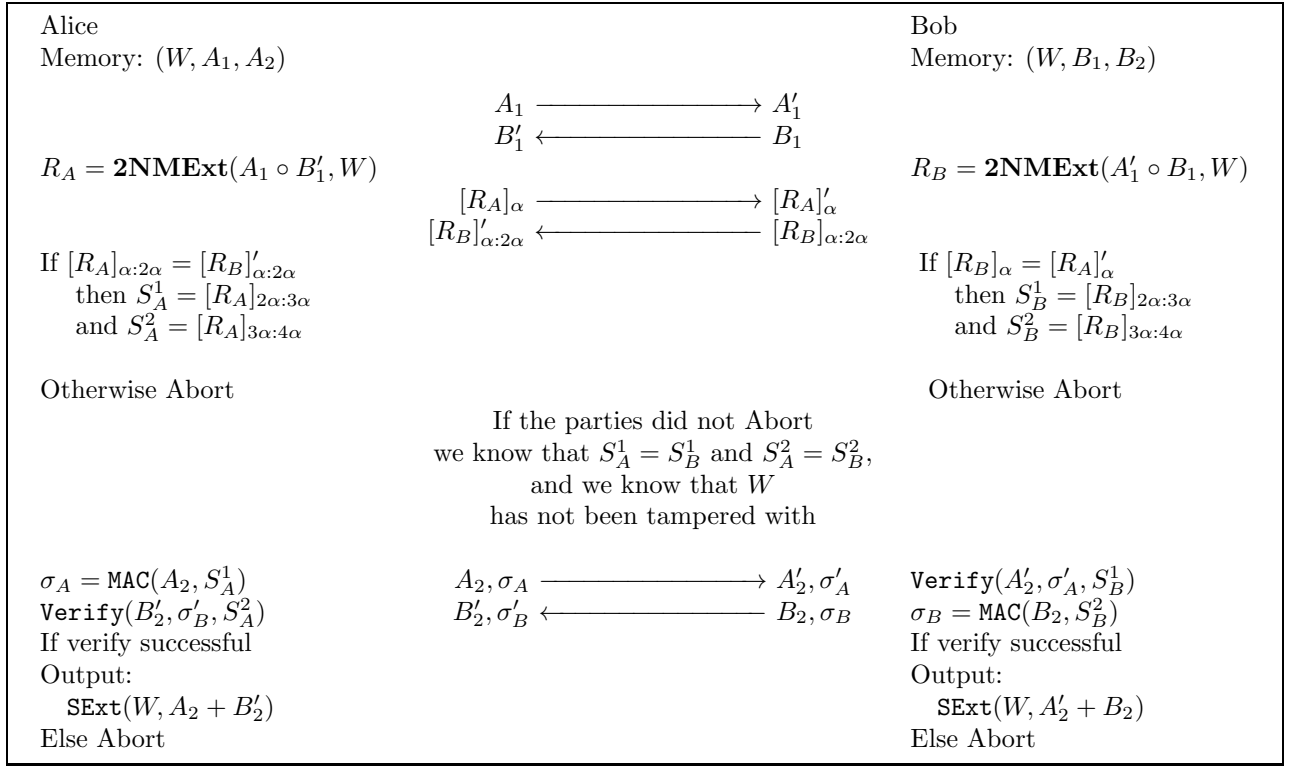

\centering
\fbox{
\begin{minipage}{16.5cm}
\begin{tabular}{lcl}
Alice & ~ &  Bob \\
Memory: $(W, A_1, A_2)$ & ~ &  
Memory: $(W, B_1, B_2)$ \vspace{0.2cm}\\

 ~ & $A_1$ $\xrightarrow{\quad \quad \quad \quad  \quad \quad \quad \quad}$ $A'_1$ ~\\
 ~ & $B'_1$ $\xleftarrow{\quad \quad \quad \quad  \quad \quad \quad \quad}$ $B_1$ ~\\
$R_A=\twonmext(A_1\concat B'_1, W)$  ~ &  ~ &  $R_B=\twonmext(A'_1\concat B_1, W)$\\
 ~ & $[R_A]_\alpha$ $\xrightarrow{\quad \quad \quad \quad  \quad \quad \quad \quad}$ $[R_A]'_\alpha$ ~\\
 ~ & $[R_B]'_{\alpha : 2\alpha}$ $\xleftarrow{\quad \quad \quad  \quad \quad  \quad \quad \quad}$ $[R_B]_{ \alpha : 2\alpha}$ ~\\
If $[R_A]_{\alpha : 2\alpha}=[R_B]'_{\alpha : 2\alpha}$ & ~ & ~If $[R_B]_{\alpha}=[R_A]'_{\alpha}$ \\
$\quad \text{  then  } S^1_A= [R_A]_{ 2\alpha: 3\alpha}$ & ~ & ~$\quad \text{  then  } S^1_B= [R_B]_{2\alpha: 3\alpha}$ \\
$\quad \text{  and  } S^2_A= [R_A]_{ 3\alpha: 4\alpha}$ & ~ & ~$\quad \text{  and  } S^2_B= [R_B]_{3\alpha: 4\alpha}$ \\
~ & & ~ \\
Otherwise Abort & ~& ~ Otherwise Abort \\
~ &If the parties did not Abort  & ~\\
~ &we know that $S^1_A=S^1_B$ and $S^2_A=S^2_B$,  & ~\\
~ &and we know that $W$ & ~\\
~ &has not been tampered with  & ~\\
~ & & ~ \\
$\sigma_A= \mac(A_2,S^1_A)$ & $A_2, \sigma_A$ $\xrightarrow{\quad \quad \quad \quad  \quad \quad \quad \quad}$ $A'_2, \sigma'_A$ & $\mathtt{Verify}(A'_2,\sigma'_A, S^1_B)$\\
 $\mathtt{Verify}(B'_2,\sigma'_B, S^2_A)$ & $B'_2, \sigma'_B$ $\xleftarrow{\quad \quad \quad \quad  \quad \quad \quad \quad}$ $B_2, \sigma_B$ & $\sigma_B=\mac(B_2,S^2_B)$\\
 If verify successful & & If verify successful\\
 Output: & & Output: \\
 \quad $\sext(W,A_2+B'_2)$ & & \quad $\sext(W,A'_2+B_2)$\\
 
 Else Abort & & Else Abort\\
\end{tabular}

\end{minipage}

}
\caption{Extension of the original PA protocol. $R$ is split into $4$ parts instead of $3$. Here $\mac$ is a standard information theoretic message authentication code (MAC). And $\sext$ is any seeded extractor. When party Aborts it stops responding and the final output is $\bot$. }
\label{fig:ITPA2}
\end{figure*}

Let us analyse the protocol described in Figure \ref{fig:ITPA2} (we copy the figure below). Let $\twonmext$ be a $[(\ell_1, k_1-2\ell_2-2\gamma-1), (2\cdot\ell_2, \ell_2-\gamma-1)  \mapsto 4\alpha \sim \epsilon]$ strong non-malleable extractor for some parameter $\gamma>0$. Let shared secret $W\in \{0,1\}^{\ell_1}$ have min-entropy $k_1$, let $A_1,A_2,B_1,B2 \in \{0,1\}^{\ell_2}$ be uniform random variables. 
If Eve is passive the security is straight forward thus we will only consider the case of active Eve. We will follow the original proof \cite{AORSS20} very closely.
Let us focus on the case where Alice is the one with corrupted memory $F(W,(A_1,A_2))= \tilde W, (\tilde A_1, \tilde A_2)$. Since randomness $(\tilde A_1,\tilde A_2)$ is controlled by the adversary we can simply reveal $(\tilde a_1, \tilde a_2)=(\tilde A_1,\tilde A_2)$ it along with original randomness $(a_1,a_2)=(A_1,A_2)$, this makes $\tilde W$ only a function of $W$, let's denote it as $\tilde W= f(W)$, moreover let us denote $B'_1=g(B_1)$. As in the original paper we define  $\cL=\{w:f(w)=w\}$ and $\cR=\{b_1:g(b)=b\}$.

In the proof of Theorem 6 in \cite{AORSS20} in point (2.b) authors prove that if $\Pr(W\notin \cL \lor B_1\notin \cR \lor a_1=\tilde a_1 ) > 2^{-\gamma}$ then $\Pr(S_B\neq \bot\;|\; W\notin \cL \lor B_1 \notin \cR)<\epsilon + 2^{-\alpha}$, thus Bob will abort.

The only case left to analyse is the point (2.a) where $W\in \cL \land B_1\in \cR \land a_1=\tilde a_1$. We assume that  
$\Pr(W\in \cL \land B_1\in \cR \land a_1=\tilde a_1) > 2^{-\gamma}$ (else this case happens with negligible probability). Authors argue that $W$ has enough entropy and thus $R_A$ is $\epsilon$ close to uniform. 
If $[R_A]'_\alpha=[R_A]_\alpha$ and $[R_B]'_{\alpha:2\alpha}=[R_B]_{\alpha:2\alpha}$, then $S^1_A\concat S^2_A=S^1_B\concat S^2_B \neq \bot$ and $S^1_A\concat S^2_A$ is $\epsilon$ close to uniform given Eve's view.
Now we know that $S^1_A\concat S^2_A=S^1_B\concat S^2_B \neq \bot$ and $\tilde W= W$ so we can follow with the analysis of the extension:
First of all the $\avgCondMinEnt{W}{A_1,A_2,\tilde A_1, \tilde A_2, W\in \cL}> k_1-2\ell_2-\gamma$ (where $|A_i|=\ell$, and $\gamma$ penalty comes from probability of the event $W\in \cL$). Now notice that by the security of MAC either $\Pr((A_2\neq A'_2 \lor B_2\neq B'_2) \land \text{ neither Alice or Bob Aborts}) <2\cdot 2^{-\Omega(\alpha)}$.

Further observe that even if Eve controls $A_2$, and $ A_2$ has no entropy and it might depend on $W$, still $B_2$ is uniform and independent of $(A_2)$. Thus $A_2+B_2$ is uniform\footnote{Technically speaking Eve can abort protocol by tampering with $A_2$ or $B_2$, Alice and Bob will simply abort. However $A_2$ and $B_2$ are no longer fully uniform conditioned on the event that Eve let them through. This is not a problem, by Lemma \ref{lma:extReduction}, this only doubles extraction epsilons.} and independent of $W$. Now we have uniform independent seed, all we have to do is extract:

Let $\sext:\{0,1\}^{\ell_1}\times\{0,1\}^{\ell_2}\rightarrow \{0,1\}^{ 0.999\cdot(k_1-2\ell_2-\gamma)}$ is a strong seeded extractor\footnote{Constant $0.999$ is just a placeholder for any constant less then $1$. By \cite{GUV09} we know that such explicit extractor exists.} with the error $2^{-\Omega(\ell_2)}$. 
Since $W$ has enough entropy  $\sext(W,A_2+B_2)$ is $2^{-\Omega(\ell_2)}$ close to uniform given the view of Eve. The analysis for Eve corrupting Bob is symmetrical. Thus we obtain the following:


\begin{theorem}
Let $\twonmext$ be a $[(\ell_1, k_1-2\ell_2-2\gamma-1), (2\cdot\ell_2, \ell_2-\gamma-1)  \mapsto 4\alpha \sim \epsilon]$ strong non-malleable extractor.
Then, there exists an $(r=6,\ell_1,k_1,2\cdot\ell_2,2\cdot \ell_2,0.999\cdot(k_1-2\ell-\gamma) ,2^{-\Omega(\ell_2)} ,\delta=\eps+2^{-\alpha}+2\cdot 2^{-\gamma}+2^{-\Omega(\alpha)})$-privacy amplification protocol against memory-tampering active adversaries.
\end{theorem}
And thus when we plug in our extractor and some example parameters we get:
\begin{corollary}
 For shared secret $W$ with $|W|=n$ and  $\minEnt{W}>0.803 \cdot n$ and $|A_i|=|B_i|=0.001n$ we get privacy amplification protocol that outputs $0.8 \cdot n$ uniform bits, and has a security  $2^{-\Omega(\sqrt{n})}$.
\end{corollary}

\bibliographystyle{splncs04}
\bibliography{writeup}

\end{document}